\documentclass[aps,twocolumn,floatfix,superscriptaddress]{revtex4-2}
\usepackage[utf8]{inputenc}
\usepackage{amsmath}
\usepackage{amsfonts}
\usepackage{amssymb}
\usepackage{amsthm}
\usepackage{mathrsfs}
\usepackage{graphicx}
\usepackage[table,dvipsnames]{xcolor}
\usepackage{url}
\usepackage{braket}
\usepackage{makecell,booktabs}
\usepackage{tikz}
\usepackage{hyperref}
\usepackage{wrapfig}
\usepackage{float}
\usepackage[export]{adjustbox}
\usepackage{tabularray}
\usepackage{caption}
\usepackage{subcaption}
\usepackage{placeins}
\usepackage{orcidlink}

\DeclareMathOperator{\re}{Re}

\DeclareMathOperator{\tr}{Tr}
\DeclareMathOperator{\ame}{AME}

\DeclareMathOperator{\ipr}{IPR_q}

\DeclareMathOperator{\ent}{S_q}
\DeclareMathOperator{\enttwo}{S_2}
\DeclareMathOperator{\entmin}{S^{min}_q}
\DeclareMathOperator{\entmintwo}{S^{min}_2}

\DeclareMathOperator{\entinf}{S_\infty}
\DeclareMathOperator{\entmininf}{S^{min}_\infty}

\DeclareMathOperator{\lpmaxnorm}{\mathcal{L}_q^{max}}
\DeclareMathOperator{\neari}{\text{Iso}}

\DeclareMathOperator{\lppca}{L_p-PCA}

\DeclareMathOperator*{\argmax}{argmax}
\DeclareMathOperator{\perm}{\mathcal{S}}

\newcommand{\rl}[1]{\left(#1\right)}

\newcolumntype{?}[1]{!{\vrule width 2pt }}
\allowdisplaybreaks

\newtheorem{definition}{Definition}
\newtheorem{proposition}{Proposition}

\begin{document}

\title{Minimal decomposition entropy and optimal representations\\ of absolutely maximally entangled states}
\author{N Ramadas \orcidlink{0000-0002-5541-0504}}
\email{ramadas.nedivarambath@krea.edu.in}

\affiliation{Center for Quantum Information, Communication, and  Computation, Indian Institute of Technology Madras, Chennai ~600036, India}
\affiliation{Division of Sciences, Krea University, Sri City 517646, India.}
\altaffiliation{Present address}

\begin{abstract}
Understanding and classifying multipartite entanglement is fundamental to quantum information processing. This work focuses on absolutely maximally entangled (AME) states, a class of highly entangled states characterized by their maximal entanglement across any bipartitions. To analyze and classify AME states, we employ the minimal decomposition entropy, defined as the minimum R\'{e}nyi entropy $S_q$ associated with the state’s decomposition over all local product bases. This quantity identifies the product bases in which the state is maximally localized, thereby yielding optimal representations for analyzing properties of AME states.

We develop an efficient algorithm for computing the minimal decomposition entropy for finite $q>1$ and compare AME and Haar-random states for \( q = 2 \) and \( q = \infty \) in qubit, qutrit, and ququad systems. For \( q = 2 \), AME states of four qutrits and ququads show lower minimal entropy than generic states, indicating sparser optimal forms. For \( q = \infty \)---related to the geometric measure of entanglement---AME states exhibit higher entanglement. The algorithm also simplifies known AME states into sparser representations, aiding in distinguishing genuinely quantum AME states from those constructible from classical combinatorial designs. Our results advance the classification of AME states and demonstrate the utility of minimal decomposition entropy as both a local unitary invariant and a tool for state simplification.
\end{abstract}

\maketitle

\section{Introduction}
Non-classical correlations in many-body systems play a crucial role in various physical phenomena, which are important for applications in quantum computing, information processing, and communication. Among these, quantum entanglement is a particularly important non-classical correlation, making its understanding and classification vital. Quantum entanglement in pure states is well understood in the case of bipartite systems. However, in systems involving more than two parties, the entanglement can be distributed various ways, making its characterization a more challenging problem. For example, in a three-qubit system, the GHZ state has genuine tripartite entanglement, while the W-state has zero tripartite entanglement. On the other hand, the W-state maximizes the total bipartite entanglement shared among all pairs of qubits. At an application level, the W-state is more robust to qubit loss compared to the GHZ state \cite{durThreeQubitsCan2000}.

{
In this work, we investigate a class of highly entangled multipartite states known as absolutely maximally entangled (AME) states \cite{helwigAbsoluteMaximalEntanglement2012}. An AME state of $N$ parties with local dimension $d$, denoted $\ame(N,d)$, is defined by the property of being maximally entangled across every possible bipartition of the system. This maximal entanglement makes AME states a powerful resource for various quantum information tasks. Absolutely maximally entangled states have applications in quantum parallel transportation and secret sharing schemes \cite{helwigAbsoluteMaximalEntanglement2012}. They are related to classical error correcting codes \cite{bernalserranoExistenceAbsolutelyMaximally2017,helwigAbsolutelyMaximallyEntangled2013a} and quantum error correcting codes\cite{scottMultipartiteEntanglementQuantumerrorcorrecting2004,raissiOptimalQuantumError2018,alsinaAbsolutelyMaximallyEntangled2021, huberQuantumCodesMaximal2020}. In holography theories and high energy physics, AME states are referred to as perfect tensors \cite{pastawskiHolographicQuantumErrorcorrecting2015}, 
and they enable the construction of simplified models of the AdS/CFT correspondence—most notably the ``HaPPY" quantum error-correcting code \cite{pastawskiHolographicQuantumErrorcorrecting2015,pastawskiCodePropertiesHolographic2017,osborneDynamicsHolographicCodes2020,mazurekQuantumErrorcorrectionCodes2020a,jahnHolographicTensorNetwork2021a,steinbergHolographicCodesHyperinvariant2023,bistronBulkboundaryCorrespondenceHyperinvariant2025}. The breadth of these applications provides the core motivation for our study of AME states. 

Despite their wide range of applications, the existence, construction, and classification of AME states remain subjects of active investigation. In general, the existence of $\ame(N,d)$ states remains open and is listed as Problem 35 in a list of open problems in quantum information \cite{OpenQuantumProblems}. It is established that qubit AME states exist only for $N=2,3,5,6$ \cite{higuchiHowEntangledCan2000, bennettMixedstateEntanglementQuantum1996,scottMultipartiteEntanglementQuantumerrorcorrecting2004,huberAbsolutelyMaximallyEntangled2017}. AME states exist for any number of parties, provided the local dimension is sufficiently large \cite{helwigAbsolutelyMaximallyEntangled2013a}. A table maintained online provides known constructions of AME states \cite{huberTableAMEStates}.

Absolutely maximally entangled states are constructed from graph states \cite{heinMultipartyEntanglementGraph2004,helwigAbsolutelyMaximallyEntangled2013} classical maximum distance separable codes \cite{helwigAbsolutelyMaximallyEntangled2013a,grasslOptimalQuantumCodes2004}, and classical combinatorial designs \cite{goyenecheAbsolutelyMaximallyEntangled2015,goyenecheEntanglementQuantumCombinatorial2018,clarisseEntanglingPowerPermutations2005}. Constructions based on classical combinatorial designs naturally lead to a distinction between classically realizable and genuinely quantum AME states. An AME state is considered genuinely quantum if it cannot be constructed from a classical combinatorial design \cite{zyczkowskiUnderstandingQuantumSolution2023}. For instance, AME states of four parties can be constructed from combinatorial designs known as orthogonal Latin squares for all local dimensions $d \geq 3$, except for $d=6$. The existence of orthogonal Latin squares of order 6, famously known as ``Euler's 36 officers problem," admits no classical solution \cite{tarryProbleme36Officiers1900}. However, quantum solutions to the problem, given by $\ame(4,6)$ states, have been found \cite{ratherThirtysixEntangledOfficers2022, ratherAbsolutelyMaximallyEntangled2023, ratherConstructionPerfectTensors2024, bruzdaTwounitaryComplexHadamard2024, grossThirtysixOfficersArtisanally2025}, and these are considered genuinely quantum.

The classification of absolutely maximally entangled (AME) states, and multipartite states more generally, can be pursued through a variety of frameworks, one of which is based on their interconvertibility. Two $N$-party pure states, $\ket{\Psi}$ and $\ket{\Psi'}$, are said to be LU equivalent if there exists a set of local unitary operators $u_1, u_2, \dots, u_N$ such that 
$\ket{\Psi'} = \left( u_1 \otimes u_2 \otimes \cdots \otimes u_N \right) \ket{\Psi}.$ 
LU-equivalent states have identical entanglement content, making them equally useful for any application. Allowing local operations along with classical communication leads to a broader classification under the local operations and classical communication (LOCC) scheme \cite{bennettConcentratingPartialEntanglement1996,bennettExactAsymptoticMeasures2000}. A coarser classification is provided by stochastic LOCC (SLOCC) \cite{durThreeQubitsCan2000,acinGeneralizedSchmidtDecomposition2000}, which considers conversions achievable with a nonzero probability of success. For AME states, however, the LU, LOCC, and SLOCC equivalence classes coincide, so it suffices to consider LU equivalence for their classification \cite{burchardtStochasticLocalOperations2020}. 

Given the utility of AME states across quantum information tasks, developing methods for their analysis and classification is essential. Classifying AME states based on their LU equivalence is challenging because their reduced density matrices are maximally mixed, which erases distinguishing local information \cite{krausLocalUnitaryEquivalence2010,sauerweinTransformationsPureMultipartite2018,burchardtStochasticLocalOperations2020}. One established approach relies on polynomial invariants \cite{grasslComputingLocalInvariants1998,rainsPolynomialInvariantsQuantum2000,kodiyalamCompleteSetNumerical2004} and has recently been applied to the classification of AME states \cite{ratherAbsolutelyMaximallyEntangled2023,ramadasLocalUnitaryEquivalence2025}. 

Here we propose an alternative based on minimal decomposition entropy, — a multipartite entanglement measure introduced by Enr'iquez \emph{et al.} \cite{enriquezMinimalRenyiIngarden2015} (also called the minimal R\'{e}nyi–Ingarden–Urbanik entropy). The motivation for using this measure is that it not only provides an LU-invariant for classification but also yields a simplified representation ideal for analysis—an advantage over polynomial invariants, which do not offer such a form. Note that approaches such as the generalized Schmidt decomposition \cite{acinGeneralizedSchmidtDecomposition2000, carteretMultipartiteGeneralisationSchmidt2000} and the higher order singular value decomposition (HOSVD) \cite{liuLocalUnitaryClassification2012,liClassificationArbitraryMultipartite2013} can also be used for classification and to obtain a simpler state representation. However, these methods are limited for AME states: HOSVD is ineffective because it relies on the diagonalization of one-party reduced density matrices—which, for AME states, are already proportional to the identity. In contrast, the generalized Schmidt decomposition is directly related to our approach—it can be viewed as a special case of the minimal decomposition entropy principle. This highlights minimal decomposition entropy as a more general and suitable framework for the analysis of AME states.

For a pure quantum state, the decomposition entropy is defined as the entropy of the probability distribution associated with its coefficients in a given orthonormal basis. When the R\'{e}nyi entropy is used, this leads to the R\'{e}nyi–Ingarden–Urbanik entropy \cite{enriquezMinimalRenyiIngarden2015}, $\ent$, defined for $q\geq 0$. It quantifies the state's localization in the chosen basis and is called participation entropy \cite{liuQuantumAlgorithmsInverse2025} in that context. Minimizing $\ent$ over all product bases yields the minimal decomposition entropy, which has been studied as a multipartite entanglement measure \cite{enriquezMinimalRenyiIngarden2015,enriquezMaximallyEntangledMultipartite2016}. It also identifies product bases in which the state is least delocalized, thereby providing an optimal representation within its LU equivalence class. Since the measure becomes increasingly sensitive to localization for larger values of $q$, lower values of $q$ are typically preferred. 

An optimal representation can provide a simpler and sparser form of the AME state, facilitating the investigation of key properties. Such sparsity is especially valuable for numerically discovered states, which often exhibit full support in the computational basis. It also aids in identifying whether the state is genuinely quantum or can be derived from a classical combinatorial design. Moreover, sparse representations reduce memory and runtime requirements when the state is used in classical simulations or computation.

Finding the minimal decomposition entropy for multipartite states is analytically challenging, necessitating efficient numerical methods. Prior work employed a random walk approach \cite{enriquezMaximallyEntangledMultipartite2016}, which can suffer from slow convergence. To overcome this, we propose an algorithm based on complex \(L_p\)-norm principal component analysis, where each step performs local optimization via gradient ascent. While our algorithm is limited to \(q > 1\) it offers significantly faster convergence, making it a practical and efficient tool for systematic studies. This algorithm cannot be used for cases $q \leq 1$ as the gradient may not exist due to non-differentiablity. We employ this algorithm to perform a systematic study of AME states, using Haar-random states as a benchmark to isolate the distinctive features from generic behavior.}

Our analysis yields several concrete outcomes. A lower bound for the minimal decomposition entropy is established for AME states. The distributions of decomposition entropy and its minimal value for $q=2$ and $q=\infty$ are computed for AME and Haar random states in qubit, qutrit, and ququad systems with varying numbers of parties. For $q=2$, it is observed that AME states of four qutrits and four ququads have smaller minimal decomposition entropy compared to Haar random states, suggesting that they can be more optimally represented. For $q=\infty$, which is related to the geometric measure of entanglement, AME states are found to be highly entangled in terms of minimal decomposition entropy or the geometric measure of entanglement, in contrast to Haar random states. We also identify states that exhibit the maximum minimal decomposition entropy for both $q=2$ and $q=\infty$. Additionally, we demonstrate the use of the algorithm to obtain simpler and sparser representations of AME states, and discuss its usefulness in determining whether a given AME state is genuinely quantum or can be constructed from classical combinatorial designs.

\section{Absolutely maximally entangled states}
In this section, we briefly review absolutely maximally entangled states. Let $\ket{\Psi}$ be a pure state of an $N$-partite quantum system described by the total Hilbert space $
\mathcal{H}_{[N]} = \mathcal{H}_1 \otimes \mathcal{H}_2 \otimes \cdots \otimes \mathcal{H}_N,$
where $[N] = \{1, 2, \ldots, N\}$. 
Each subsystem $\mathcal{H}_i$ is assumed to have the same dimension $d$. 
In the computational basis 
\(
\big\{ \ket{j_1, j_2, \ldots, j_N} : j_i \in \{0, 1, \ldots, d-1\} \big\},\)
the state $\ket{\Psi}$ can be expressed as
\begin{equation}\label{eq:multipartitestate}
\ket{\Psi} = \sum_{j_1, j_2, \ldots, j_N = 0}^{d-1} 
c_{j_1, j_2, \ldots, j_N} \ket{j_1, j_2, \ldots, j_N},
\end{equation}
where $c_{j_1, j_2, \ldots, j_N} \in \mathbb{C}$ are the complex amplitudes satisfying the normalization condition \[ \sum_{j_1, j_2, \ldots, j_N = 0}^{d-1} 
|c_{j_1, j_2, \ldots, j_N}|^2 = 1.\]

The reduced density matrix $\rho_S$ of a subsystem $S$ consisting of $k$ parties is obtained by tracing out the remaining parties:
\begin{equation}
\rho_S = \tr_{S^c} \big( \ket{\Psi}\bra{\Psi} \big),
\end{equation}
where $S^c$ denotes the complement of $S$ in $[N]$. 
For the subsystem composed of the first $k$ parties, the reduced density matrix can be expressed as
\begin{equation}
\rho_{[k]} = \frac{1}{d^k} A A^\dagger,
\end{equation}
where $[k] = \{1,2,\ldots,k\}$, and $A$ is a $d^k \times d^{N-k}$ matrix defined by its elements $(A)^{j_1,j_2,\ldots,j_k}_{j_{k+1},\ldots,j_N}$
\begin{align}\label{eq:matricization}
\begin{aligned}
= \braket{j_1,j_2,\ldots,j_k|A|j_{k+1},\ldots,j_N}
= \sqrt{d^{k}}\, c_{j_1,j_2,\ldots,j_N}.
\end{aligned}
\end{align}

The reduced density matrix $\rho_S$ for any other subsystem $S$ of $k$ parties can be obtained by reshaping the matrix $A$. 
Let the subsystem $S$ be written as $S = \{\sigma(1), \sigma(2), \ldots, \sigma(k)\}$, 
where $\sigma$ is a permutation of the $N$ subsystem indices. 
The permutation corresponding to a given subset $S$ is not unique; however, since the remaining parties are traced out, 
their ordering under $\sigma$ is irrelevant.

The action of the permutation $\sigma$ on the $N$-partite state yields
\begin{align}
\begin{aligned}
P_{\sigma}\ket{\Psi} 
= \frac{1}{\sqrt{d^k}} 
\sum_{j_1,j_2,\ldots,j_N = 0}^{d-1} 
(A^{R_\sigma})^{j_{\sigma(1)},j_{\sigma(2)},\ldots,j_{\sigma(k)}}_{j_{\sigma(k+1)},\ldots,j_{\sigma(N)}}
\\
\times \ket{j_{\sigma(1)},j_{\sigma(2)},\ldots,j_{\sigma(N)}},
\end{aligned}
\end{align}
where $P_{\sigma}$ is the permutation operator acting as 
$P_{\sigma}\ket{j_1,j_2,\ldots,j_N} = \ket{j_{\sigma(1)},j_{\sigma(2)},\ldots,j_{\sigma(N)}}$.
After applying this permutation, the subsystem $S$ occupies the first $k$ positions and its reduced density matrix is given by
\begin{equation}
\rho_S = \frac{1}{d^k} A^{R_\sigma} (A^{R_\sigma})^\dagger,
\end{equation}
where
\begin{equation}\label{eq:reshaping}
(A^{R_\sigma})^{j_{\sigma(1)},j_{\sigma(2)},\ldots,j_{\sigma(k)}}_{j_{\sigma(k+1)},\ldots,j_{\sigma(N)}} 
= (A)^{j_1,j_2,\ldots,j_k}_{j_{k+1},\ldots,j_N}.
\end{equation}
These reshaping operations can be viewed as generalizations of the realignment and partial transpose operations~\cite{chenMatrixRealignmentMethod2003,peresSeparabilityCriterionDensity1996}.

The subsystem linear entropy, defined as \[E_L(\rho_S) = 1 - \tr \rl{ \rho_S^2},\] quantifies the entanglement between a subsystem $S$ and its complement $S^c$. It ranges from $0$ for product states to $1-1/d^{|S|}$ for maximally entangled states, where $|S|$ denotes the cardinality of $S$. If the subsystem $S$ is maximally entangled with its complement $S^c$, the reduced density matrix $\rho_S$ is proportional to identity operator, that is, $\rho_S \propto \mathbb{I}_{d^{|S|}} $, assuming $|S| \leq |S^c|$. Requiring maximal entanglement for all such bipartitions defines a class of highly entangled states known as $k$-uniform states, and, in the special case where $k = \lfloor N/2 \rfloor$, 
as absolutely maximally entangled states.
\begin{definition}
An N-party state $\ket{ \Psi} \in \mathcal{H}_{[N]} $ is $k$-uniform if the reduced density matrix $\rho_S =
\frac{1}{d^{|S|}} \mathbb{I}_{d^{|S|}}$ for any subset $S \subset [N] $ with $|S| \leq k \leq \lfloor N/2 \rfloor $.
If $k=\lfloor N/2 \rfloor$, the state is an absolutely maximally entangled state, denoted by $\ame(N,d)$.
\end{definition}

If $\ket{\Psi}$ is a $k$-uniform state, then for every subsystem $S$ with $|S| = k \le \lfloor N/2 \rfloor$,
the reduced density matrix satisfies $\rho_S = \frac{1}{d^{|S|}} \mathbb{I}_{d^{|S|}} 
= \frac{1}{d^{|S|}} A^{R_\sigma} (A^{R_\sigma})^\dagger,$
for all permutations $\sigma \in \perm_N$. 
Hence, each reshaped matrix $A^{R_\sigma}$ is an isometry. 
Conversely, if all reshaped matrices $A^{R_\sigma}$ are isometries for every $\sigma \in \perm_N$, 
then the state $\ket{\Psi}$ is $k$-uniform.
For even $N$ and $k = N/2$, corresponding to absolutely maximally entangled (AME) states, 
the matrix $A$ is referred to as a \emph{multiunitary matrix}~\cite{goyenecheAbsolutelyMaximallyEntangled2015}.
In practice, it is unnecessary to consider all possible $N!$ permutations. 
If $A^{R_\sigma}$ is an isometry, then so is $A^{R_{\sigma'}}$ 
for any $\sigma'$ related to $\sigma$ by independent rearrangements of the first $k$ and last $N-k$ indices, i.e., $\sigma' = (\sigma_1 \times \sigma_2) \circ \sigma,
\qquad \text{where } \sigma_1 \in \perm_k,~ \sigma_2 \in \perm_{N-k}$.
Therefore, it suffices to consider permutations from the quotient set 
$\perm_N / (\perm_k \times \perm_{N-k})$. 
The size of this quotient set is $\binom{N}{k}$, 
so only $\binom{N}{k}$ distinct permutations — or equivalently, 
distinct bipartitions — need to be considered.

Since $A^{R_\sigma}$ is an isometry, its rank is $d^k$. 
This implies that the minimum number of nonzero coefficients of the state $\ket{\Psi}$ 
in the computational basis---that is, its support---is at least $d^k$. 
$k$-uniform states that achieve this bound are called \emph{minimal support $k$-uniform states}. 
In the case of AME states, the subclass of minimal support AME states is of particular interest 
due to their connection with maximum distance separable (MDS) codes~\cite{bernalserranoExistenceAbsolutelyMaximally2017}.

Classical combinatorial designs such as Latin squares, orthogonal Latin squares and orthogonal arrays \cite{goyenecheAbsolutelyMaximallyEntangled2015,goyenecheEntanglementQuantumCombinatorial2018,clarisseEntanglingPowerPermutations2005} are used to construct AME states. A Latin square (LS) of order $d$ is a $d\times d$ array filled with $d$ distinct symbols, each occuring exactly once in every row and column. Two Latin squares of the same order are said to be orthogonal if, when superimposed, every ordered pair of symbols occurs exactly once. Orthogonal arrays (OA) are generalization of orthogonal Latin squares (OLS). An orthogonal array denoted by OA$(r,N,d,k)$ is an $r \times N$ array of $d$ distinct symbols such that in every $r\times k$ array each $k$-tuple of $d$ symbols appears exactly the same number of times. A Latin square of order $d$ is an OA$(d^2,3,d,1)$ and a pair of Latin squares of order $d$ is an OA$(d^2,4,d,2)$. The following are examples of a Latin square, a pair of orthogonal Latin squares, and an orthogonal array:
\begin{align*}
\begin{aligned}
\text{LS}(3) &=\left[ \begin{array}{ccc}
0&1&2 \\ 1&2&0 \\ 2&0&1 \\ 
\end{array} \right],\\
\text{OLS}(3) &= \left[
\begin{array}{ccc}
(0,0)&(1,1)&(2,2) \\  (1,2)&(2,0)&(0,1) \\ (2,1)&(0,2)&(1,0) \\ 
\end{array}\right],\\
\text{OA}(9,4,3,2) &= \begin{bmatrix}0&0&0&0\\
0&1&1&1\\
0&2&2&2\\
1&0&1&2\\
1&1&2&0\\
1&2&0&1\\
2&0&2&1\\
2&1&0&2\\
2&2&1&0\\
\end{bmatrix} .
\end{aligned}
\end{align*}
%
To give an example of the construction, the following is an AME state of four qutrits obtained from the $\text{OLS}(3)$, or equivalently from the $\text{OA}(9,4,3,2)$, presented above:
\begin{equation}\label{eq:ame43}
\begin{split}
\ket{\Psi_{4,3}}=& \frac{1}{3}\left(\ket{0000}+\ket{0111}+\ket{0222}+\ket{1012}\right. + \ket{1120} \\ &\left.+\ket{1201}+\ket{2021}+\ket{2102}+\ket{2210}\right).
\end{split}
\end{equation} 

The construction based on orthogonal Latin squares shows the existence of four party AME states for all $d \geq 3$, with the notable exception of $d=6$. The existence of orthogonal Latin squares of order 6, famously known as ``Euler's 36 officers problem", admits no classical solution \cite{tarryProbleme36Officiers1900}. A quantum solution the problem, given by an $\ame(4,6)$ state, was first given in \cite{ratherThirtysixEntangledOfficers2022}. Subsequent work demonstrated the existence of an infinite number of LU inequivalent solutions \cite{ratherAbsolutelyMaximallyEntangled2023}, and further studies have obtained new solutions \cite{ratherConstructionPerfectTensors2024,bruzdaTwounitaryComplexHadamard2024} including an explicit analytical construction \cite{grossThirtysixOfficersArtisanally2025}. 

An $\ame(4,6)$ state is considered \textit{genuinely quantum} in the sense that it cannot be constructed from a classical combinatorial designs \cite{zyczkowskiUnderstandingQuantumSolution2023}. In contrast, the $\ame(4,3)$ state given in Eq.~\ref{eq:ame43} is not genuinely quantum, as it can be derived from a classical OLS. It has been shown there is only one LU equivalence class of $\ame(4,3)$ states \cite{ratherAbsolutelyMaximallyEntangled2023}. This implies that genuine quantum $\ame(4,3)$ states do not exist. 

An $\ame(4,d)$ state gives rise to an $d\times d$ array of quantum states that serves as an analog of a pair of orthogonal Latin squares, known as a quantum orthogonal Latin square (QOLS). The QOLS associated with the $\ame(4,3)$ state Eq.~\ref{eq:ame43} consists of separable states. On the other hand, any QOLS of order 6 corresponding to an $\ame(4,6)$ state necessary involves entangled states. This generalization of classical designs has led to the development of  quantum combinatorial designs including quantum Latin squares \cite{mustoQuantumLatinSquares2016}, quantum orthogonal Latin squares, and quantum orthogonal arrays \cite{goyenecheEntanglementQuantumCombinatorial2018}. These quantum designs may offer solution to combinatorial problems for which classical solutions do not exist. In this context, a key problem is to determine whether a given a AME state, up to LU equivalence, can be derived from a classical combinatorial design or is genuinely quantum. Identifying optimal representation of an AME state within the LU equivalence class is helpful for addressing this question.

\section{Minimal decomposition entropy of AME states}\label{sec:maxipr}
In this section, we begin by briefly reviewing the minimal decomposition entropy introduced in \cite{enriquezMinimalRenyiIngarden2015}. Next, we discuss some special cases of this entropy. For AME states, a lower bound for the decomposition entropy is established. Finally, we propose a numerical algorithm for computing the minimal decomposition entropy of multipartite states.

Consider a general multipartite state $\ket{\Psi} \in \mathcal{H}_{[N]}$ of an $N$-party system as given in Eq.~\ref{eq:multipartitestate}. The set $\lbrace p_{j_1,j_2,...,j_N} = |c_{j_1,j_2,...,j_N}|^2 \rbrace$ constructed using the coefficients forms a probability distribution. The decomposition entropy is defined as the R\'{e}nyi entropy associated with this probability distribution \cite{enriquezMinimalRenyiIngarden2015}: 
\begin{align}
\ent( \ket{\Psi} ) = \frac{1}{1-q} \log \left( \sum_{j_1,j_2,...,j_N=0}^{d-1} p^q_{j_1,j_2,...,j_N} \right), \quad q \geq 0.
\end{align}
This entropy has been referred to as the R\'{e}nyi-Ingarden-Urbanik entropy in \cite{enriquezMinimalRenyiIngarden2015}. The minimal decomposition entropy of a state $\ket{\Psi}$ is defined by minimizing $\ent( \ket{\Psi} )$ over all local unitary operators:
\begin{align}
\entmin(\ket{\Psi}) = \min \limits_{u_1,u_2,\dots,u_N \in \mathbb{U}(d)} \ent \left( (u_1 \otimes u_2 \otimes \cdots \otimes u_N) \ket{\Psi} \right),
\end{align}
where $u_1, u_2, \dots, u_N \in \mathbb{U}(d)$ are local unitary matrices. By definition, the minimal decomposition entropy is an LU invariant, as it remains constant for all states within an LU equivalence class.

The following are special cases of the minimal decomposition entropy.

\begin{enumerate}
\item $q=0$~:~This case corresponds to minimal Hartley entropy, which represents the minimum support that a state can achieve within its LU equivalence class. It has been shown that for any $N$-party state, $Nd(d-1)/2$ coefficients can always be set to zero through suitable local unitary transformations \cite{carteretMultipartiteGeneralisationSchmidt2000}. Since $\ent$ is a non-increasing function of $q$, this provides an upper bound on the minimal decomposition entropy: $\entmin \leq \log R_{\text{max}}$, where  $R_{\text{max} }= d^N - Nd(d-1)/2$. 
\item $q=1$~:~In the limit $q\to 1$, the minimal Shannon entropy is obtained. This represents the minimum entropy of outcomes over all choices of local projective measurements \cite{bravyiEntanglementEntropyMultipartite2003}.
\item $q=\infty$~:~ In the limit $q\to \infty $, the minimal decomposition entropy can be expressed as  
\begin{align}\label{eq:sinf_gme}
\begin{aligned}
S_\infty^{\text{min}} ( \ket{\Psi})  = - \log (\lambda) ,
\end{aligned}
\end{align}
 where 
\begin{align*}
\begin{aligned}
\lambda = \max_{\ket{\phi_{\text{sep}}}}|\braket{\Psi |\phi_{\text{sep}}}|^2
\end{aligned}
\end{align*}
is obtained by optimizing over all fully separable states $\ket{\phi_{\text{sep}}} = \ket{\phi_1} \otimes \ket{\phi_2} \otimes\cdots \otimes \ket{\phi_N}$. The quantity $E_G(\ket{\Psi}) = 1-\lambda$ defines the geometric measure of entanglement (GME) \cite{weiGeometricMeasureEntanglement2003}, which is one measure of multipartite entanglement. This establishes the relation between $S_\infty^{\text{min}} $ and the GME: 
\begin{align}
\begin{aligned}\label{eq:gme_sinfinity}
E_G(\ket{\Psi}) = 1- e^{-S_\infty^{\text{min}} ( \ket{\Psi}) }.
\end{aligned}
\end{align}
\end{enumerate}

For $q>1$, the minimal decomposition entropy is directly related to the inverse participation ratio (IPR), a commonly used measure of localization, defined as
\begin{align}
\begin{aligned}
\ipr (\ket{\Psi}) = \sum_{j_1,j_2,...,j_N=0}^{d-1} p^q_{j_1,j_2,...,j_N}.
\end{aligned}
\end{align}
The functional relation between the decomposition entropy and the IPR is given by
\begin{align}
\begin{aligned}
\ent\rl{\ket{\Psi}} = \frac{1}{1-q}\log  \rl{  \ipr (\ket{\Psi}) }.
\end{aligned}
\end{align}
The IPR and decomposition entropy have been widely employed in diverse contexts, including as measures of localization in many-body systems~\cite{Thouless74, Kramer93, Evers08, Rodriguez09, Rodriguez10}, to probe multifractality and extract fractal dimensions~\cite{halsey1986fractal, Stanley88, dziurawiec2023unraveling, Garcia09, Backer19}, to characterize quantum phases ~\cite{Stephan09, Stephan10, Stephan14, Luitz14universal, Luitz14improving, Pino17, Lindinger19, Pausch21,PhysRevB.109.064302, Lozano21}, to detect ergodicity breaking and many-body localization~\cite{de2013ergodicity}, to study measurement-induced phase transitions and circuit dynamics~\cite{ sierant2022universal, Sierant22measurement}. 

\begin{proposition}
The minimal decomposition entropy of an $\ame(N,d)$ state is bounded below by $ \lfloor N/2 \rfloor \log(d) $ for all values of $q$. For $q<\infty$, the lower bound is achieved if and only if the AME state is of minimal support.
\end{proposition}
\begin{proof}
We will prove this proposition by establishing a lower bound for the decomposition entropy for various values of $q$. For $q=0$, the minimal Hartley entropy $\text{S}^{\text{min}}_0 \geq \lfloor N/2 \rfloor \log(d)$. This follows from the fact that the support of any $\ame(N,d)$ state is greater than or equal to $d^{\lfloor N/2 \rfloor}$. The remaining cases are considered now.

Let $\ket{\Psi}$ be an $\ame(N,d)$ state with a decomposition as given in Eq.~\ref{eq:multipartitestate}, and $A$ be the a $d^k \times d^{(N-k)}$ matrix associated with it, where $k\leq \lfloor N/2\rfloor$.
For simplicity, we introduce multi-index notation $\mu = \left\{i_1,i_2,...,i_k\right\}$ and $\nu = \left\{i_{k+1},...,i_N\right\}$. In terms of the matrix elements $A^{\mu}_{\nu}$, the decomposition entropy is given by:
\begin{align}
\begin{aligned}
\ent(\ket{\Psi}) = \frac{1}{1-q} \log\rl{\frac{1}{d^{kq}} \sum_{\mu,\nu} |A^{\mu}_{\nu}|^{2q} }.
\end{aligned}
\end{align}
Since the matrix $A$ is an isometry, we have $\sum_{\nu} |A^{\mu}_{\nu}|^{2} = 1$ for each fixed $\mu$. Therefore, the set $ \left\{|A^{\mu}_{\nu}|^{2}\right\} $ for a fixed $\mu$ forms a probability distribution. We consider different values of $k$ and bipartitions to achieve a tighter bound; the decomposition entropy does not depend on these choices.

For $q=\infty$, the decomposition entropy $S_\infty(\ket{\Psi}) = -\log \rl{\text{max} ( |A^{\mu}_{\nu}|^{2} )  } + k \log(d)$. Since $\sum_{\nu} |A^{\mu}_{\nu}|^{2} = 1$ for any $\mu$, $|A^{\mu}_{\nu}|^{2} \leq 1$. This implies, $ \text{max} ( |A^{\mu}_{\nu}|^{2} )  \leq 1$. Therefore   $S_\infty \geq k \log (d)$. A tighter bound is obtained when $k = \lfloor N/2 \rfloor$. Therefore, $S_\infty \geq  \lfloor N/2 \rfloor \log (d)$.

The case $q=1$ corresponds to Shannon entropy:
\begin{align}
\begin{aligned}
\text{S}_1( \ket{\Psi}) = -\sum_{\mu}\rl{ \sum_\nu \frac{1}{d^{k}}|A^{\mu}_{\nu}|^2 \log \rl{ \frac{1}{d^{k}} |A^{\mu}_{\nu}|^2 }}. 
\end{aligned}
\end{align}
Simplifying the above expression, we get
\begin{align}
\begin{aligned}
\text{S}_1( \ket{\Psi}) &=  - \frac{1}{d^{k}} \sum_{\mu}\rl{ \sum_\nu | A^{\mu}_{\nu}|^2 \log \rl{ |A^{\mu}_{\nu}|^2 }} + k \log(d)\\
& =  \frac{1}{d^{k}} \sum_{\mu} \text{S}_1(\pmb{A}_\mu )  + k \log(d),
\end{aligned}
\end{align}
where $ \text{S}_1(\pmb{A}_\mu ) $ denotes the Shannon entropy of the probability distribution corresponding to the $\mu$-th row of the matrix $A$. The minimum value of the first term is zero if the Shannon entropy corresponding to each row of the matrix vanishes. This can occur only when each row of matrix $A$ contains exactly one non-zero entry of magnitude 1. This possibility happens if $k=\lfloor N/2 \rfloor$ and the corresponding AME state has minimal support. Consequently, the lower bound of $\text{S}_1$ is $ \lfloor N/2 \rfloor \log (d) $, which is attained when the AME state has minimal support. These arguments hold for any subsystem of $\lfloor N/2 \rfloor$ parties or equivalently for any reshaping of the matrix $A$.

For values $0<q<\infty$, excluding the value $q=1$, we can write the decomposition entropy as 
\begin{align}
\begin{aligned}
\ent(\ket{\Psi})&  = \frac{1}{1-q} \log\rl{\frac{d^k}{d^{kq}} \times \frac{1}{d^k} \rl{  \sum_{\mu,\nu} |A^{\mu}_{\nu}|^{2q} }  }\\& = \frac{1}{1-q} \log\rl{\frac{1}{d^{k}} \sum_{\mu} \sum_{\nu} |A^{\mu}_{\nu}|^{2q} } + k \log (d)
\end{aligned}
\end{align}
Let $x_\mu = \sum_{\nu} |A^{\mu}_{\nu}|^{2q} $. Since $\sum_{\nu} |A^{\mu}_{\nu}|^{2} = 1$, $x_\mu$ will not vanish for any values of $q < \infty$. Therefore, we can use the logarithmic form of the arithmetic-geometric mean inequality:
\begin{align}
\begin{aligned}
\log \rl{ \frac{\sum_\mu x_\mu}{d^k}} \geq \frac{1}{d^k}  \sum_\mu \log \rl{x_\mu}.
\end{aligned}
\end{align}
This leads to:
\begin{align}
\begin{aligned}
\ent(\ket{\Psi})& \geq \frac{1}{d^{k}} \ent(\pmb{A}_\mu) + k \log (d),
\end{aligned}
\end{align}
where $\ent(\pmb{A}_\mu)$ is the R\'{e}nyi entropy of the probability distribution corresponding to the $\mu$-th row of the matrix $A$. By similar arguments as in the $q=1$ case, we obtain $\ent \geq \lfloor N/2 \rfloor \log(d)$, where the lower bound attained only for minimal support AME states. Combining all the results, for $\ame(N,d)$ states, we have
\begin{align}
\begin{aligned}
\entmin \geq \lfloor N/2 \rfloor \log(d),~~ \text{for all} ~q.
\end{aligned}
\end{align}
\end{proof}

For $q = \infty$, the lower bound can be attained if at least one row of $A$ contains an non-zero entry of magnitude 1, without imposing any constraints on the other rows. Consequently, the lower bound can be achieved even with a non-minimal support AME state. For example, the $\ame(4,4)$ state $\ket{O_{16}}$, corresponding to a 2-unitary matrix $O_{16}$, which has been shown to have non-minimal support \cite{ratherConstructionLocalEquivalence2022}, still attains the minimal value $\text{S}_\infty^{\text{min}} = \log(16) = 2 \log(4)$. In Fig.\ref{fig:O16}, a 2-unitary matrix that is LU equivalent to $O_{16}$ is given, which has non-zero entry of magnitude 1 in one of the rows.

\subsection{Algorithm to find minimal decomposition entropy for $q>1$}
In bipartite systems, the minimal decomposition entropy can be computed using the Schmidt coefficients \cite{enriquezMinimalRenyiIngarden2015}. However, for multipartite systems, no known analytical techniques are available to find the minimal decomposition entropy. Therefore, numerical methods are employed. In this section, we propose an algorithm to compute the minimal decomposition entropy for $q>1$. The algorithm leverages local optimization and gradient ascent, offering advantages over the random walk method used in \cite{enriquezMinimalRenyiIngarden2015}.

The $L_p$-norm of a state $\ket{X}$ is given by $  \rVert X \lVert_{p} = \rl{\sum_i |x_i|^p}^{1/p}$, where $x_i = \braket{i|X}$ are the coefficients in the basis $\left\{\ket{i}\right\}$. The $L_p$ norm a matrix is similarly defined: $\lVert A \rVert_p = \rl{\sum_{i,j} |A_{i,j}|^p}^{1/p}$. Expressing the decomposition entropy in terms of the $L_p$-norm of $\ket{\Psi}$ for $p=2q$ gives
\begin{align}\label{eq:entropy_l_p}
\begin{aligned}
\ent \rl{\ket{\Psi}}  = \frac{1}{1-q} \log \rl{\rl{\rVert \ket{\Psi}  \lVert_{2q}}^{2q} }. 
\end{aligned}
\end{align}
Similarly, the minimal decomposition entropy
\begin{align}
\begin{aligned}
&\entmin \rl{\ket{\Psi}}  =\\
&\min_{ u_1,\cdots,u_N \in \mathbb{U}(d)}\frac{1}{1-q} \log  \rl{  \rl{\rVert (u_1\otimes \cdots \otimes u_N) \ket{\Psi}  \lVert_{2q}}^{2q}   }.
\end{aligned}
\end{align}
For $q>1$, this problem is equivalent to the maximization problem
\begin{align}
\begin{aligned}
\lpmaxnorm ( \ket{\Psi}  ) = \max_{ u_1,\cdots,u_N \in \mathbb{U}(d)} \rl{  \rVert (u_1 \otimes\cdots \otimes u_N) \ket{\Psi}  \lVert_{2q}   }^{2q}.
\end{aligned}
\end{align}
The core of our algorithm is to find $\lpmaxnorm ( \ket{\Psi}  ) $. For a fixed set of matrices $\{u^{(0)}_2,..,u^{(0)}_N\}$, the above maximization problem is $\lpmaxnorm \rl{ \ket{ \Psi}} \big\rvert_{u^{(0)}_2,..,u^{(0)}_N} $
\begin{align}
\begin{aligned}
= \max_{ u_1 \in \mathbb{U}(d)}  \rl{  \rVert (u_1 \otimes u^{(0)}_2 \otimes\cdots \otimes u^{(0)}_N) \ket{\Psi}  \lVert_{2q}   }^{2q}.
\end{aligned}
\end{align}

We observe that $$  \rl{\rVert (u_1 \otimes u^{(0)}_2 \otimes\cdots \otimes u^{(0)}_N) \ket{\Psi}  \lVert_{2q}}^{2q}    =  \rl{\rVert (u_1 M_1   \lVert_{2q}}^{2q}~~, $$ where the matrix $M_1$ of order $d \times d^{N-1}$ defined by its elements
 $$\braket{l_1|M_1|l_2,...,l_N} = \braket{l_1,l_2,...,l_N |(\mathbb{I}_d\otimes u^{(0)}_2 \otimes\cdots \otimes u^{(0)}_N)|\Psi}. $$ Therefore,
\begin{align}
\begin{aligned}\label{eq:lppca}
\lpmaxnorm \rl{ \ket{ \Psi}} \big\rvert_{u^{(0)}_2,..,u^{(0)}_N}= \max_{ u_1 \in \mathbb{U}(d)}  \rl{  \rVert (u_1 M_1  \lVert_{2q}   }^{2q}.
\end{aligned}
\end{align}
This maximization problem is related to the principle component analysis (PCA) of complex matrices using the $L_p$-norm, referred to as complex $\lppca$. The $\lppca$ problem for real matrices was investigated in \cite{kwakPrincipalComponentAnalysis2014}. For our purpose, the method is extended to the case of complex matrices. A detailed complex $\lppca$ algorithm is given in appendix \ref{app:complexlppca}. 

The complex $\lppca$ problem enables us to find a unitary matrix $u^{(1)}_1$ such that the norm in Eq.~\ref{eq:lppca} is maximized. That is, 
\begin{align}
\begin{aligned}
(u^{(1)}_1)^\dagger = \argmax_{ \theta \in \mathbb{U}(d)}  \rl{  \rVert (\theta^\dagger M_1  \lVert_{2q}   }^{2q}.
\end{aligned}
\end{align}
Following a similar approach, the matrices $ u^{(j+1)}_i $
are found using the complex $\lppca$ algorithm for $i=1,2,..,N$, and each iteration step $j=0,1,2,...$. Once the algorithm converges, the final step involves applying the obtained local unitary matrices to the original state in order to compute the decomposition entropy.  The detailed algorithm is provided below: \\
\rule{\columnwidth}{1pt}
Algorithm\\ 
\rule{\columnwidth}{1pt}
Initialize random unitary matrices $u^{(0)}_1,u^{(0)}_2,...,u^{(0)}_N \in \mathbb{U}(d)$

For $j=1,2,...$
\begin{itemize}
\item[] For $i=1,2,..,N$\\
$$(u^{(j+1)}_i)^\dagger  = \argmax\limits_{\theta  \in \mathbb{U}(d)} \rl{\lVert  \theta^\dagger M_i \rVert_{2q}}^{2q}, $$

where $M_i$ is a $d\times d^{N-1}$ matrix with entries \[\braket{l_i|M_i|l_1,\ldots,l_{i-1},l_{i+1},\ldots,l_N} = \braket{l_1,\ldots,l_N | Q_i|\Psi} \]
where $Q_i = ( u_1 \otimes \cdots  u_{i-1} \otimes \mathbb{I}_d \otimes u_{i+1} \otimes \cdots u_N )$.
\end{itemize}
\rule{\columnwidth}{1pt}
The maximization at each step ensures that the algorithm converges. However, this does not guarantee reaching the global maximum, as a local maximization is performed in each step. Multiple runs of the algorithm with randomly chosen seed unitary matrices may be necessary to find the global maximum.

\section{Numerical results}
In this section we present numerical results, primarily focusing on the cases $q=2$ and $q=\infty$. The distribution of $\enttwo,\entmintwo, \entinf,$ and $\entmininf$ for both typical states and AME states are computed for various values of $N$ and $d$. 
We also discuss examples of optimal representations of AME states, highlighting their usefulness in obtaining sparse representations and in determining whether a given AME state is genuinely quantum.

\setlength{\tabcolsep}{4pt}
\renewcommand{\arraystretch}{1.2}


\begin{table*}
\centering
\begin{tabular}{lcccccccc}
\hline
&$\braket{\enttwo}_H$ & $\braket{\entmintwo}_H$ &$\braket{\enttwo}_{\text{AME}}$ & $\braket{\entmintwo}_{\text{AME}}$ &$\braket{\entinf}_H$ & $\braket{\entmininf}_H$ &$\braket{\entinf}_{\text{AME}}$ & $\braket{\entmininf}_{\text{AME}}$\\
\hline
Three qubits& 1.531 & 0.648&1.710& 0.693 &1.113&0.379&1.381&0.693\\ 
Four qubits&2.160&1.200& &&1.589&0.694&&\\
Three qutrits&2.651&1.384&2.782&1.480 &1.965&0.852&2.169&1.304\\
Four qutrits&3.719&2.531&3.875&2.197&2.812&1.506&3.123&2.197\\
Four ququads&4.857&3.639&4.963&3.341&3.751&2.176&3.995&2.896\\
\hline
\end{tabular}
\caption{Average values of decomposition entropy $\enttwo$ and its minimal value $\entmintwo$ of Haar random states and random AME states (denoted by subscripts H and AME, respectively) for various cases. For four qubits, AME states do not exist, and the absence of values is denoted by blank entry. }
\label{tab:average}
\end{table*}

\begin{table*}
\begin{tabular}{lcccccl}
\hline
&Ensemble&Ensemble&Known&Ensemble&Ensemble&Known\\
& $\max\rl{\entmintwo}_H$ & $\max\rl{\entmintwo}_{\text{AME}}$  &  $\max\rl{\entmintwo}$ \cite{enriquezMinimalRenyiIngarden2015,enriquezMaximallyEntangledMultipartite2016} & $\max\rl{\entmininf}_H$ & $\max\rl{\entmininf}_{\text{AME}}$ &  $\max\rl{\entmininf}$ \cite{enriquezMaximallyEntangledMultipartite2016,steinbergFindingMaximalQuantum2024}\\
\hline
Three qubits&$1.257^*$&0.693&1.108&0.768&0.693& $\log(9/4) \approx0.810$\\ 
Four qubits&1.747&&$\log(6)\approx 1.791$&1.175&&$\log(9/2) \approx 1.504$\\
Three qutrits&1.789&$1.944^*$&$\log(6)\approx 1.791$&1.275&1.529&  $\log(6) \approx1.791$\\
Four qutrits&$2.893^*$&2.197&$\log(12)\approx 2.484$&1.843&2.197&$\log(9) \approx2.197$\\
Four ququads&$3.904^*$&3.518& $\dagger$&2.431&2.985&$\log(24) \approx3.178$\\
\hline
\end{tabular}
\caption{Maximum values of minimal decomposition entropy $\entmintwo$ and $\entmininf$ for Haar random states and random AME states (denoted by subscripts H and AME, respectively), computed over different ensembles. Previously conjectured or numerically found maximum values are provided for comparison. For four qubits, AME states do not exist, and the absence of values is denoted by a blank entry. Entries marked with ${}^*$ indicate newly obtained maximum values, and those marked with $\dagger$ indicate values that are not known. }
\label{tab:max}
\end{table*}

\hfill \\
\subsection{Minimal decomposition entropy for $q=2$}

The probability distributions of the estimations of $\enttwo $ and $\entmintwo$ for Haar random states and random AME states are obtained by constructing ensembles of such states. The ensemble of AME states is produced by applying the algorithm described in Appendix \ref{app:ame_algorithm} 
The minimal decomposition entropy, $\entmintwo$, is calculated using the algorithm proposed in this work. Results for qubit, qutrit, and ququad systems are presented in Figs. \ref{fig:mde_s_2}, and the corresponding ensemble averages of are listed in Table \ref{tab:average}.

\begin{figure*}
\begin{subfigure}{0.45\textwidth}
\includegraphics[width=\linewidth]{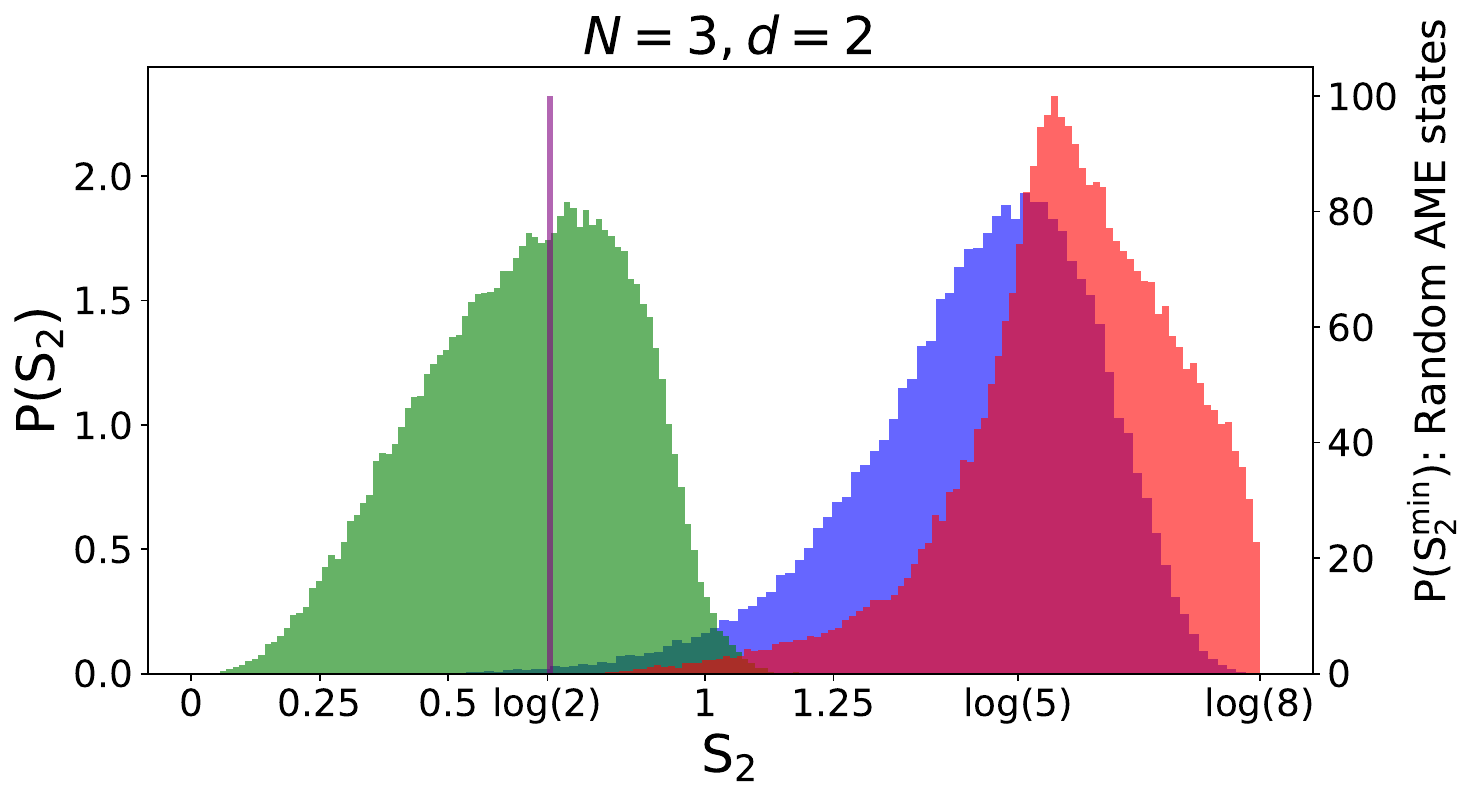}
\caption{}
\label{fig:n3d2}
\end{subfigure}
\hfill
\begin{subfigure}{0.45\textwidth}
\includegraphics[width=\linewidth]{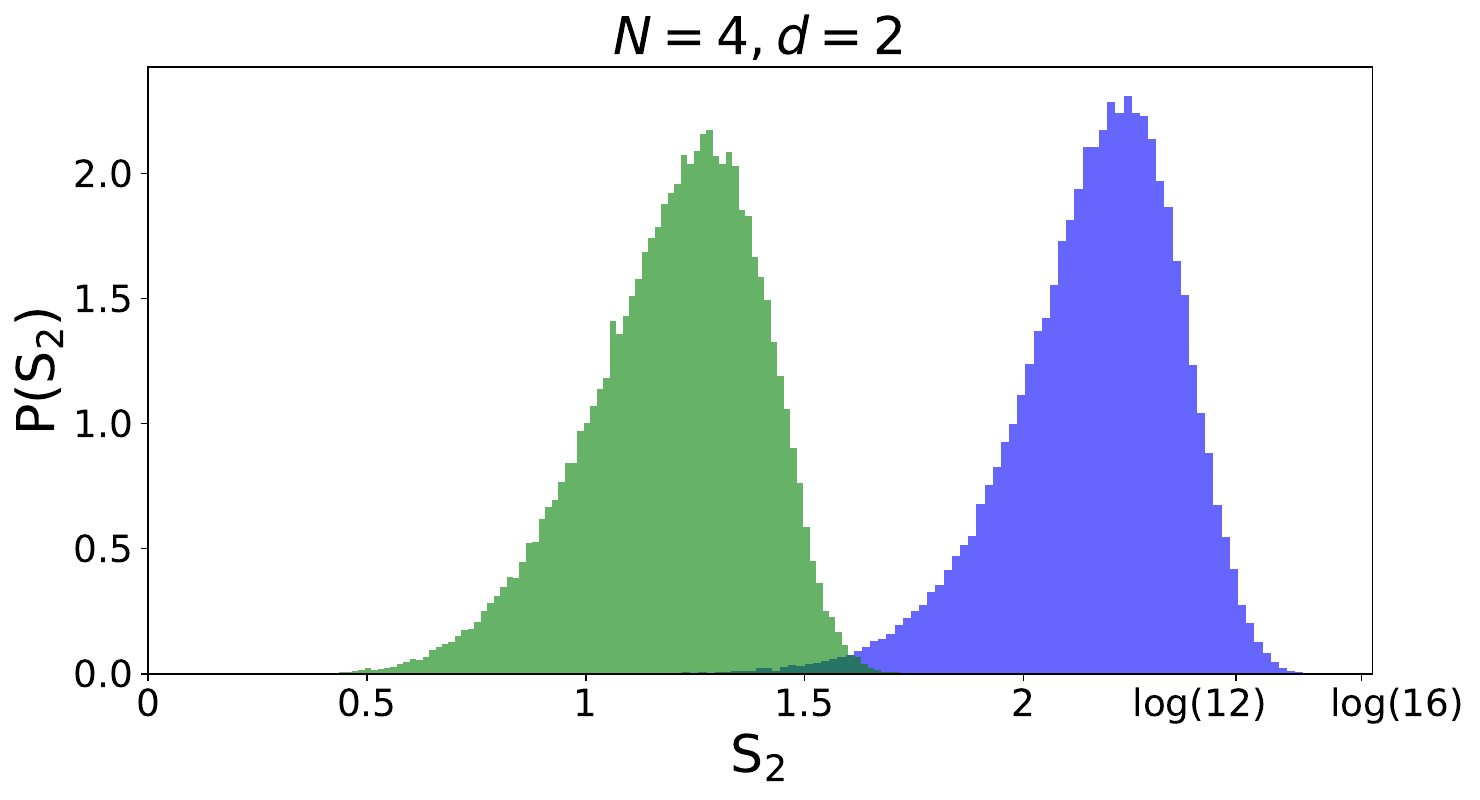}
\caption{}
\label{fig:n4d2}
\end{subfigure}

\begin{subfigure}{0.45\textwidth}
\includegraphics[width=\linewidth]{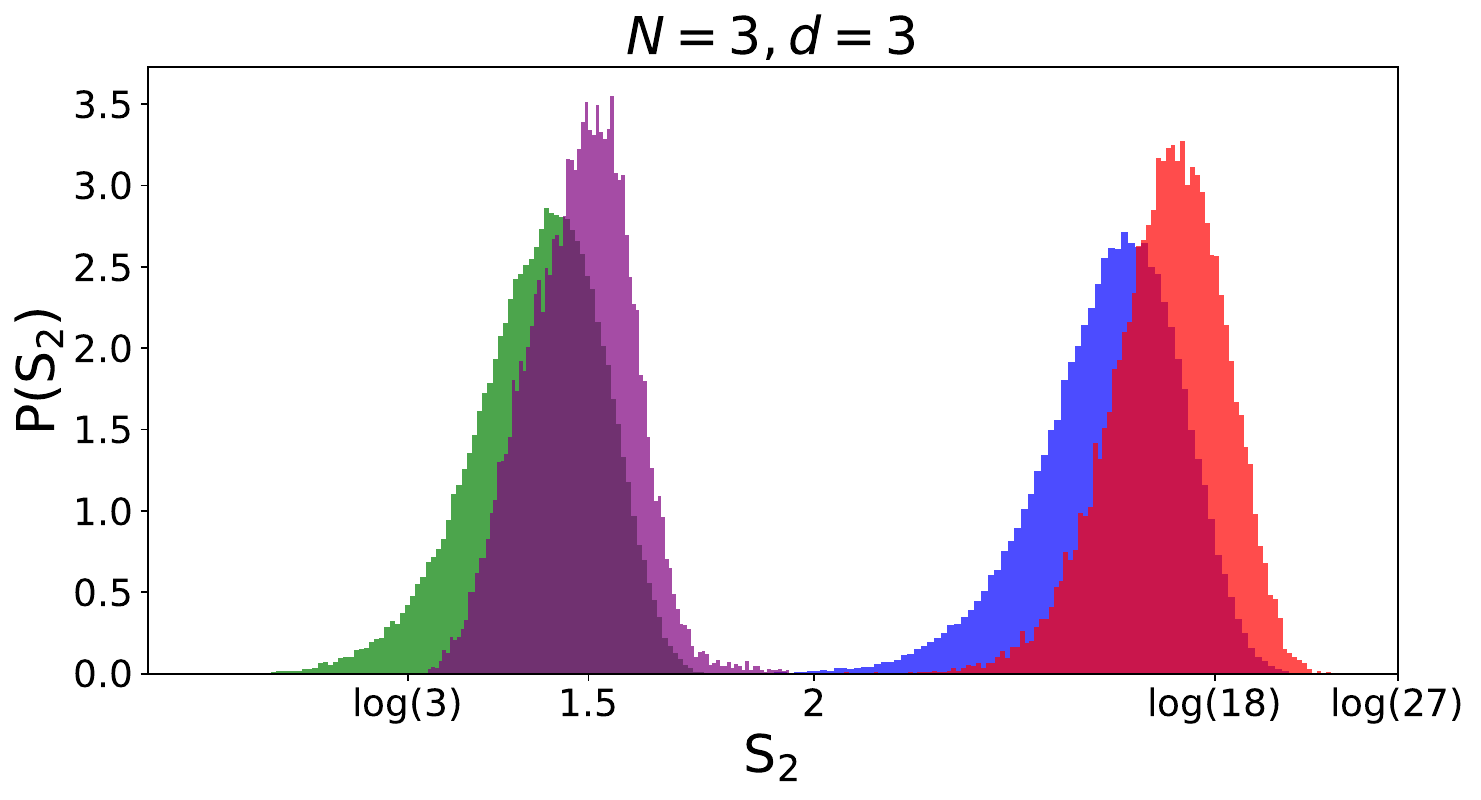}
\caption{}
\label{fig:n3d3}
\end{subfigure}
\hfill
\begin{subfigure}{0.45\textwidth}
\includegraphics[width=\linewidth]{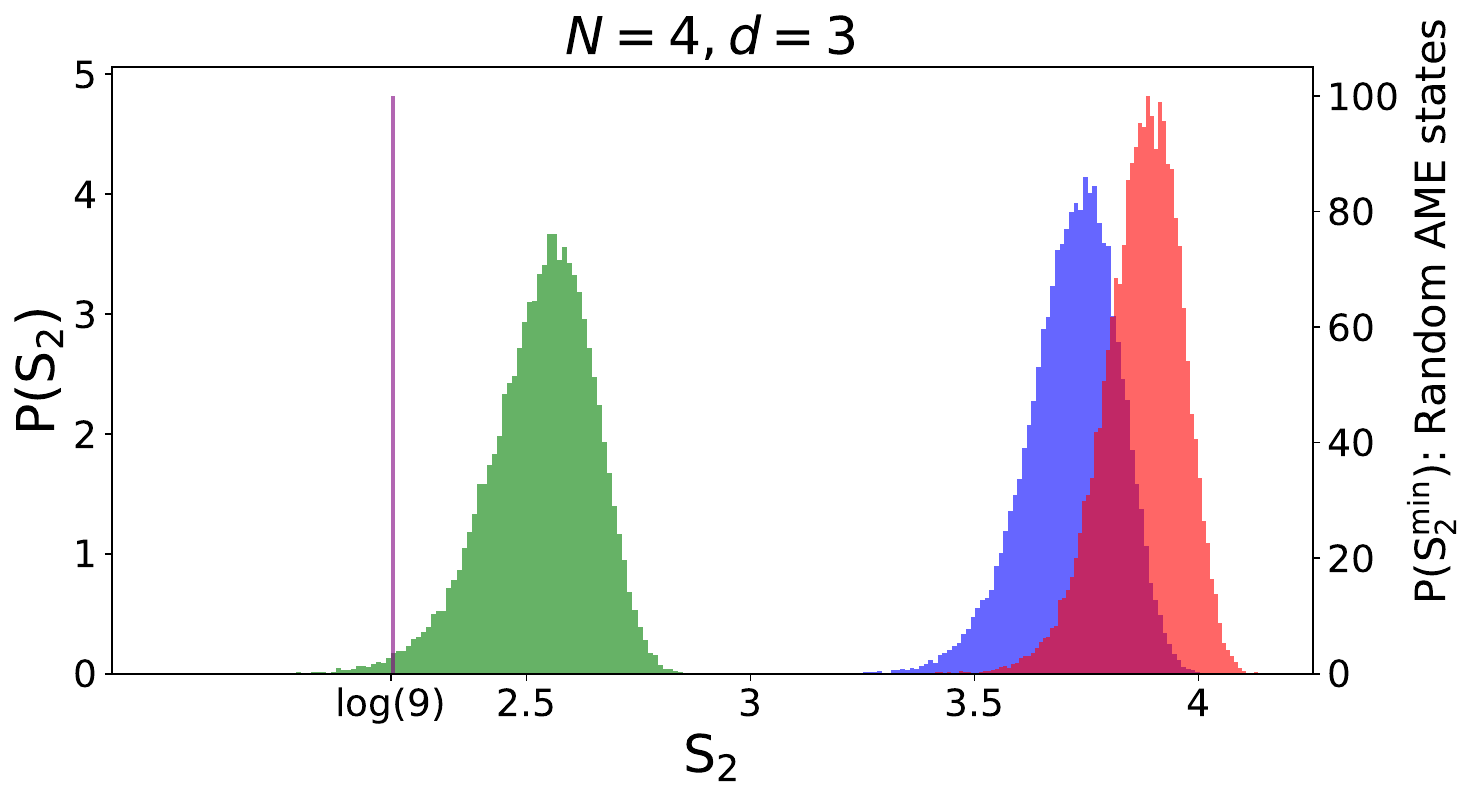}
\caption{}
\label{fig:n4d3}
\end{subfigure}

\begin{subfigure}{0.55\textwidth}
\includegraphics[width=\linewidth]{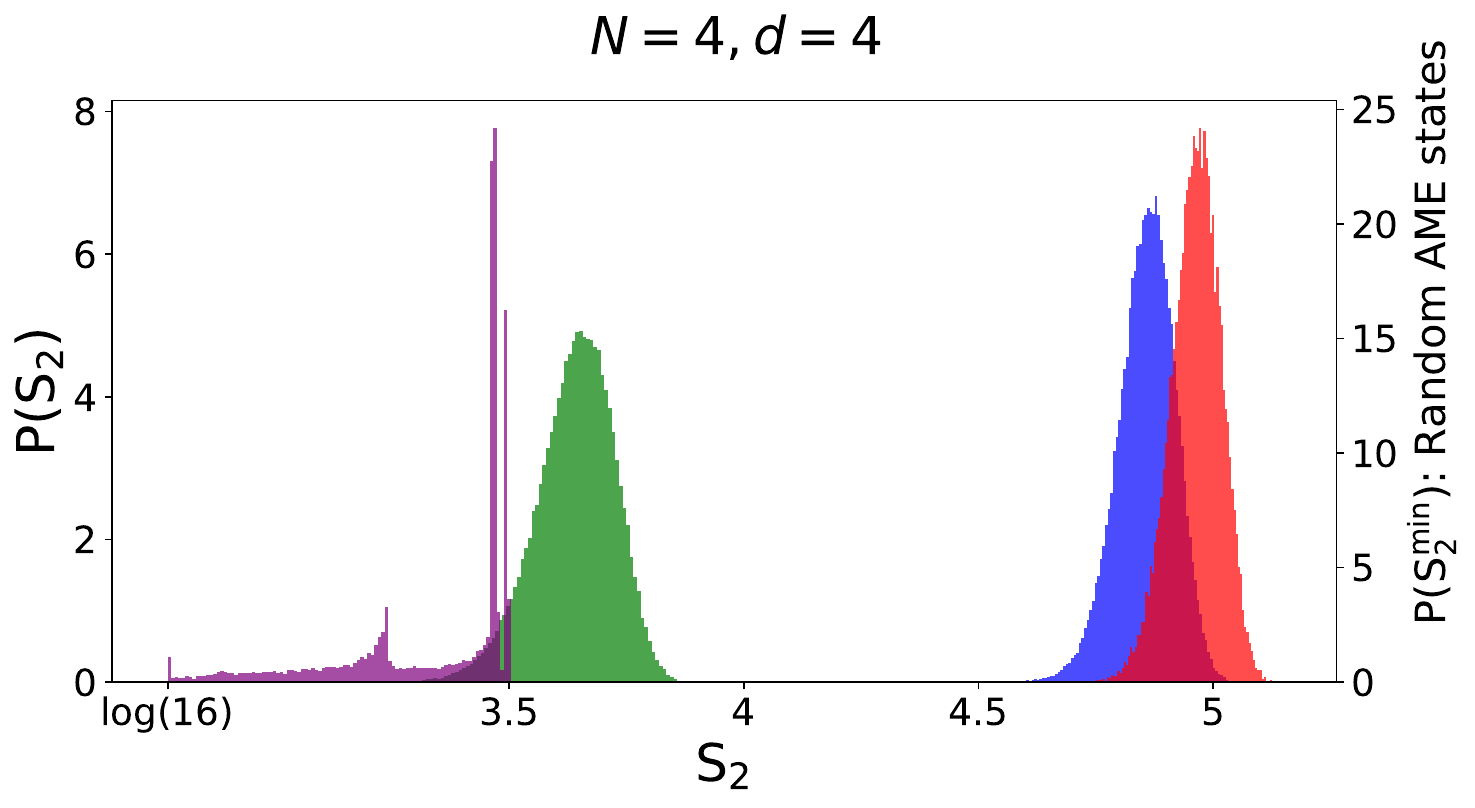}
\caption{}
\label{fig:n4d4}
\end{subfigure}
\hfill
\begin{subfigure}{0.4\textwidth}
\includegraphics[width=\linewidth]{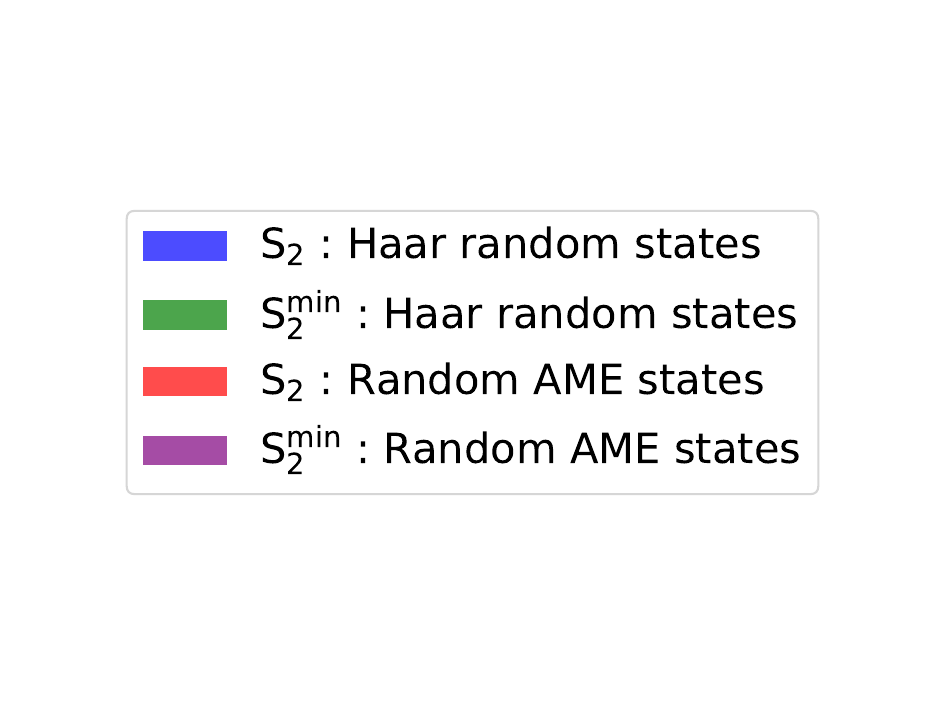}
\end{subfigure}
\caption{
Distribution of $\enttwo$ and $\entmintwo$ for Haar random states and random AME states in (a) three qubits, (b) four qubits, (c) three qutrits, (d) four qutrits, and (e) four ququads. For the ensembles of AME states, the lower bound of $\enttwo$ for $\ame(N,d)$ states, given by $\log\left(d^{\lfloor N/2 \rfloor}\right)$, is indicated. The upper bound for $\enttwo$, $\log(d^{N})$, and the known upper bound for $\entmintwo$, $\log\left(d^{N} - \tfrac{1}{2} N d(d-1)\right)$, are included wherever possible. The ensembles of $\ame(3,3)$, $\ame(4,3)$, and $\ame(4,4)$ consist of $2.5 \times 10^4$ states, while all other Haar random and AME state ensembles contain $10^5$ states.
}
\label{fig:mde_s_2}
\end{figure*}

\begin{figure*}
\begin{subfigure}{0.45\textwidth}
\includegraphics[width=\linewidth]{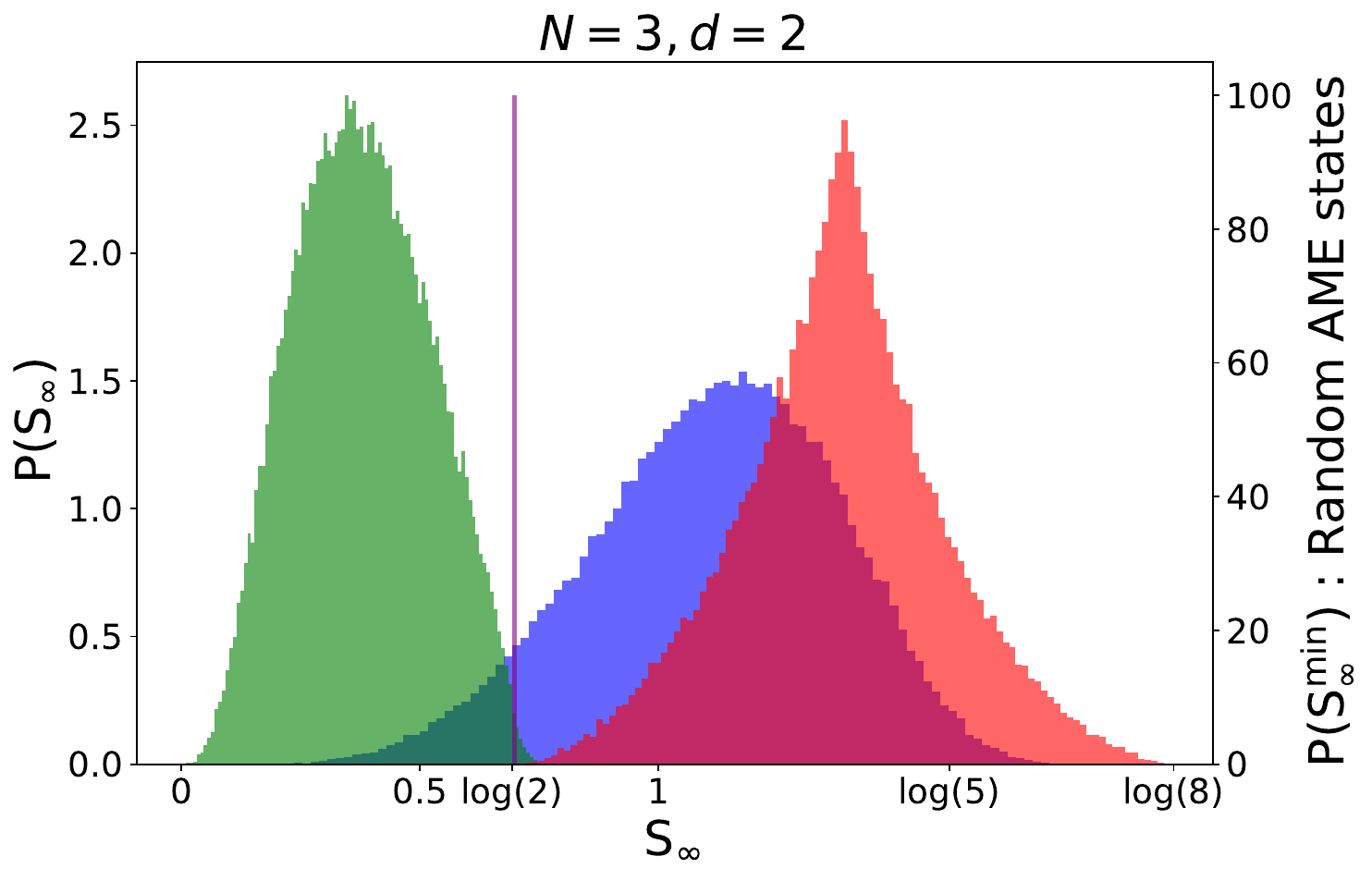}
\caption{}
\label{fig:n3d2_sinfinity}
\end{subfigure}
\hfill
\begin{subfigure}{0.45\textwidth}
\includegraphics[width=\linewidth]{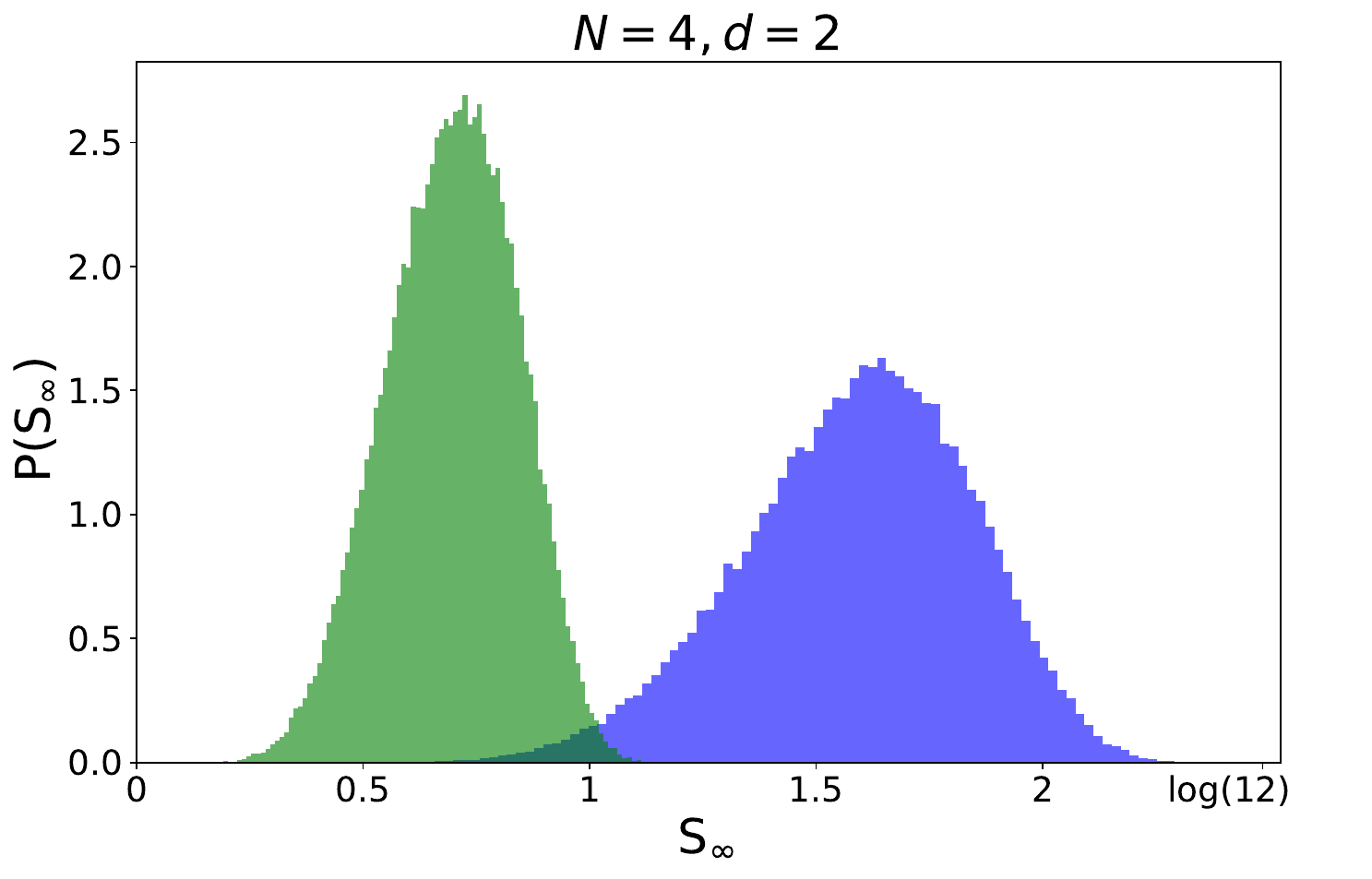}
\caption{}
\label{fig:n4d2_sinfinity}
\end{subfigure}

\begin{subfigure}{0.45\textwidth}
\includegraphics[width=\linewidth]{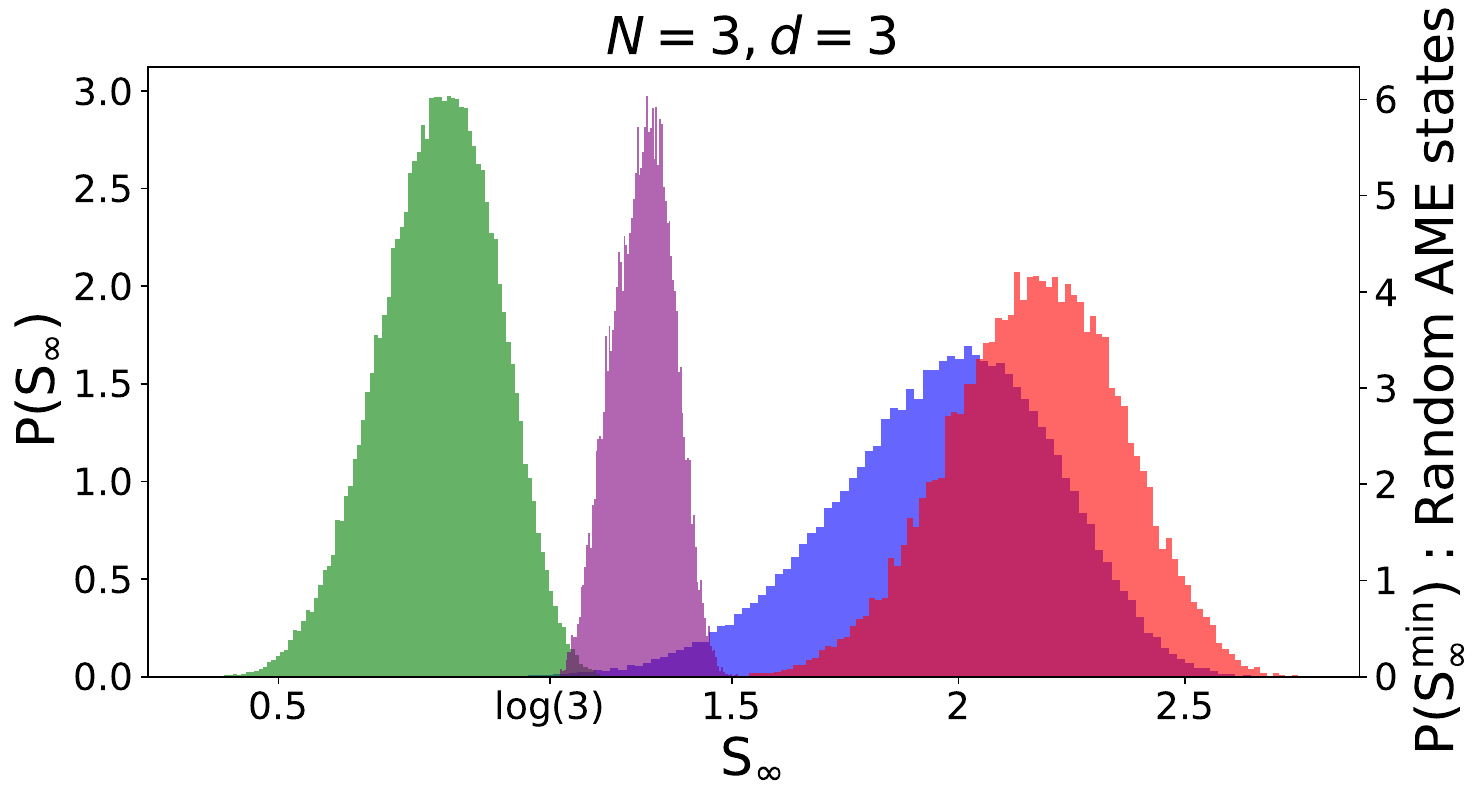}
\caption{}
\label{fig:n3d3_sinfinity}
\end{subfigure}
\hfill
\begin{subfigure}{0.45\textwidth}
\includegraphics[width=\linewidth]{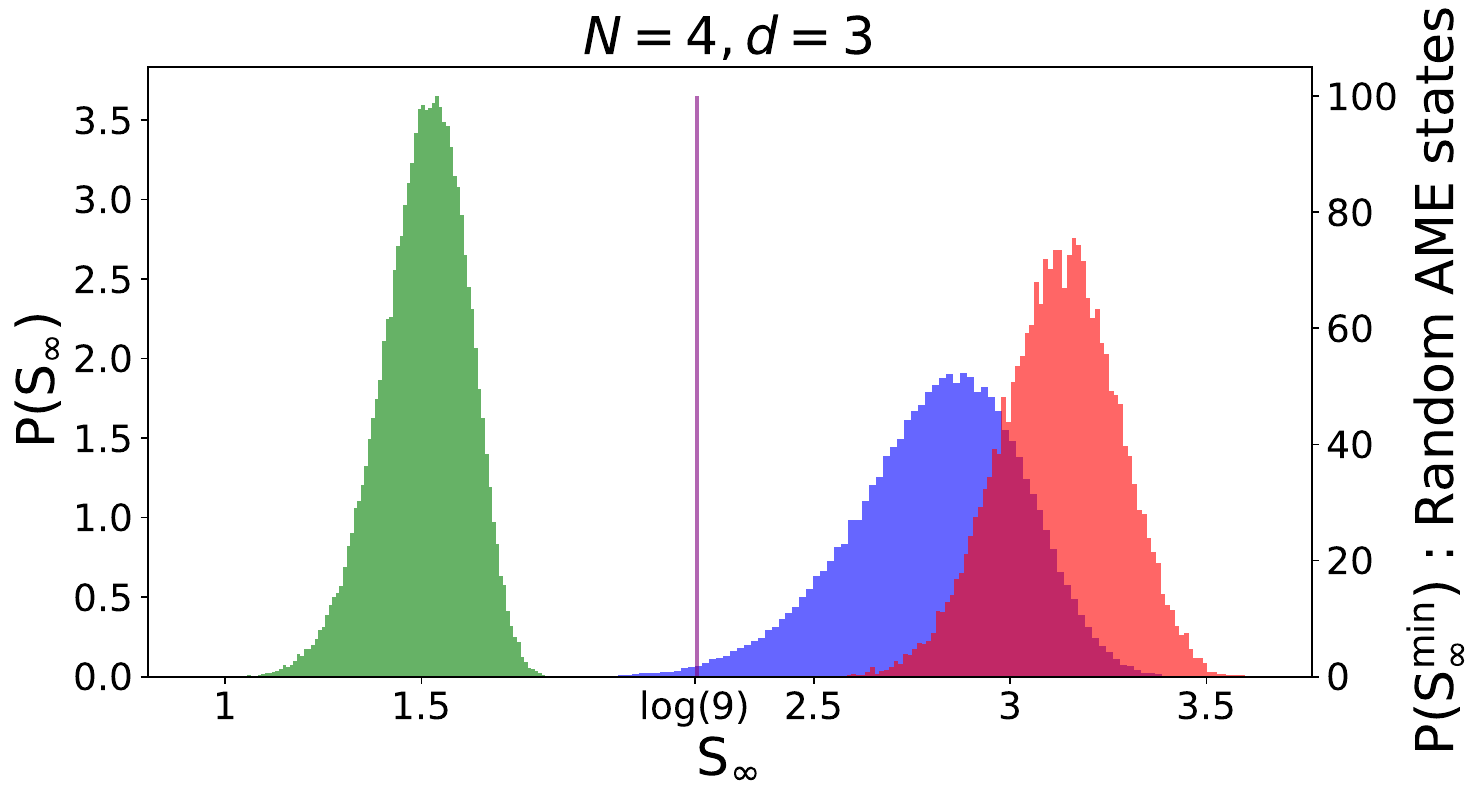}
\caption{}
\label{fig:n4d3_sinfinity}
\end{subfigure}

\begin{subfigure}{0.55\textwidth}
\includegraphics[width=\linewidth]{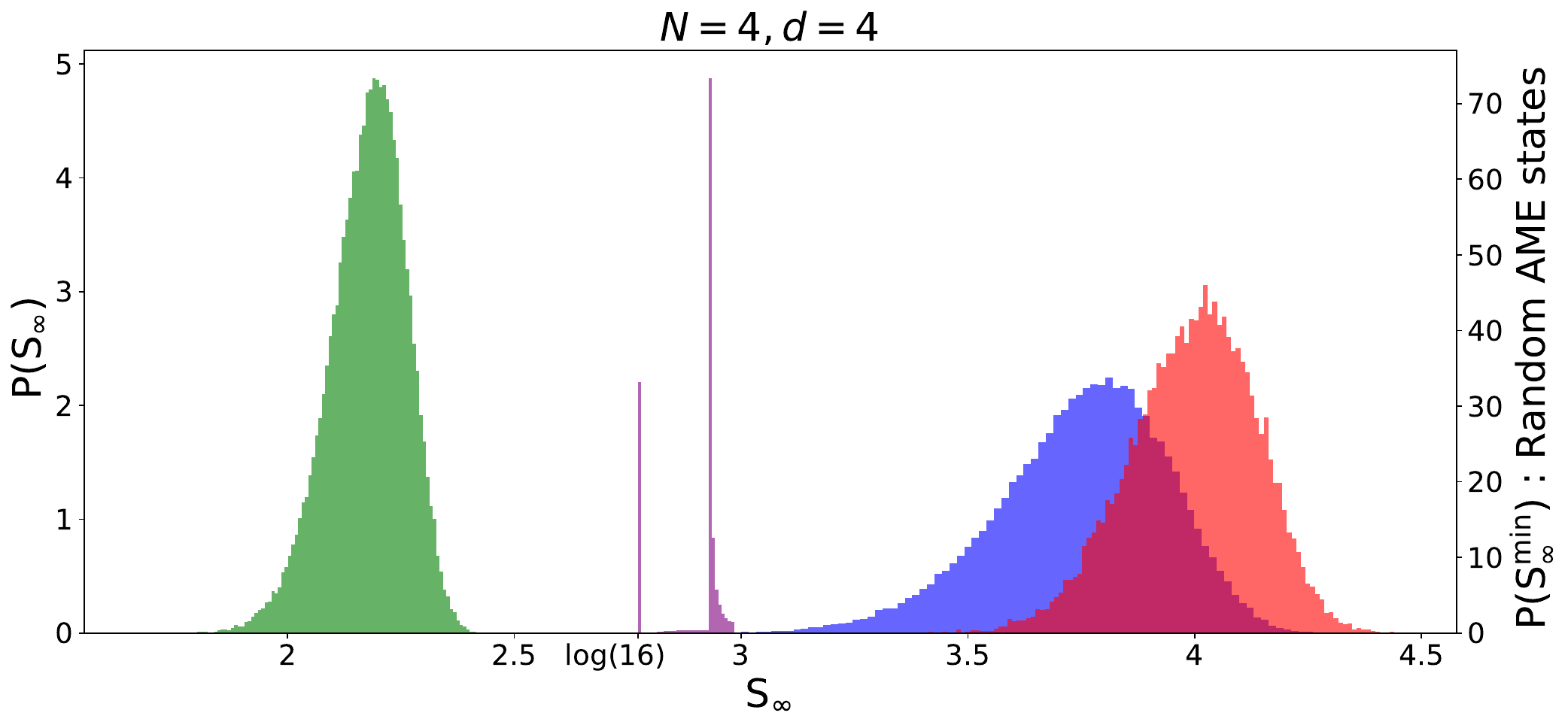}
\caption{}
\label{fig:n4d4_sinfinity}
\end{subfigure}
\hfill
\begin{subfigure}{0.4\textwidth}
\includegraphics[width=\linewidth]{images/legend.pdf}
\end{subfigure}
\caption{
Distribution of $\entinf$ and $\entmininf$ for Haar random states and random AME states in (a) three qubits, (b) four qubits, (c) three qutrits, (d) four qutrits, and (e) four ququads. For the ensembles of AME states, the lower bound of $\enttwo$ for $\ame(N,d)$ states, given by $\log\left(d^{\lfloor N/2 \rfloor}\right)$, is indicated. The upper bound for $\enttwo$, $\log(d^{N})$, and the known upper bound for $\entmintwo$, $\log\left(d^{N} - \tfrac{1}{2} N d(d-1)\right)$, are included wherever possible. The ensembles of $\ame(3,3)$, $\ame(4,3)$, and $\ame(4,4)$ consist of $2.5 \times 10^4$ states, while all other Haar random and AME state ensembles contain $10^5$ states.
}
\label{fig:mde_s_inf}
\end{figure*}

Some special cases are noted first. It is known that all $\ame(3,2)$ states are LU equivalent to the GHZ state \cite{durThreeQubitsCan2000}, which is given by
\begin{align*}
\begin{aligned}
\ket{GHZ} = \frac{1}{\sqrt{2}} \rl{ \ket{000} + \ket{111} }.
\end{aligned}
\end{align*}
Since the GHZ state has minimal support, the minimal decomposition entropy $\entmin$ for any $\ame(3,2)$ state is $ \log(2)$ for all $q$. As a result, the distribution of $\entmintwo$ of $\ame(3,2)$ states is a delta function at $\log(2)$. Similarly, all $\ame(4,3)$ states are LU equivalent to the minimal support state given in Eq.~\ref{eq:ame43} \cite{ratherAbsolutelyMaximallyEntangled2023}, and thus the minimal decomposition entropy for any $\ame(4,3)$ state is $\entmin = \log(9)$. To demonstrate this, the convergence of the algorithm for twenty random $\ame(4,3)$ states is plotted in Fig.~\ref{fig:ame_43_convergence}. The algorithm provides a set of four local unitary matrices that connect a given random $\ame(4,3) $ state to a minimal support state. For four qubits, no AME states exists \cite{higuchiHowEntangledCan2000}, and the data for Haar random states is presented in Fig.~\ref{fig:n4d2}, which shows the distribution of $\entmintwo$ for typical states.

We next examine the distribution of the minimal decomposition entropy $\entmintwo$. In the ensemble of Haar random three qubit states, approximately 45\% of states are found to have $\entmintwo \geq \log(2)$. For three qutrit states, the distribution of $\entmintwo$ of Haar random states and $\ame(3,3)$ states shows significant overlap. However, in the case of four qutrits and four ququad system, majority of Haar random states in the ensemble exhibit larger $\entmintwo$ values compared to AME states. This observation suggests that, in these cases, Haar random states are more entangled than AME states, while AME states can be transformed into a more optimal form compared to typical Haar random states.

For AME$(4,4)$ states, the distribution of $\entmintwo$ is highly irregular. This irregularity may reflect the underlying structure of the set of LU-equivalence classes, or it may be an artifact of the state generation algorithm. Further investigation, potentially involving analytical methods, is required to resolve this behavior.

Extremal states that maximize $\entmintwo$ are of particular interest when investigating maximally entangled states with respect to this measure. This topic has been explored in previous works \cite{enriquezMinimalRenyiIngarden2015, enriquezMaximallyEntangledMultipartite2016}. In Table~\ref{tab:max}, we present data on the maximum values of $\entmintwo$ observed in both the ensemble of Haar random states and AME states. We have identified states that exhibit higher $\entmintwo$ values than those previously known. These states are detailed in \cite{N_Minimal_decomposition_entropy}, and the corresponding values are provided in Table~\ref{tab:max}.

\begin{figure*}
\centering
\includegraphics[width=0.9\textwidth]{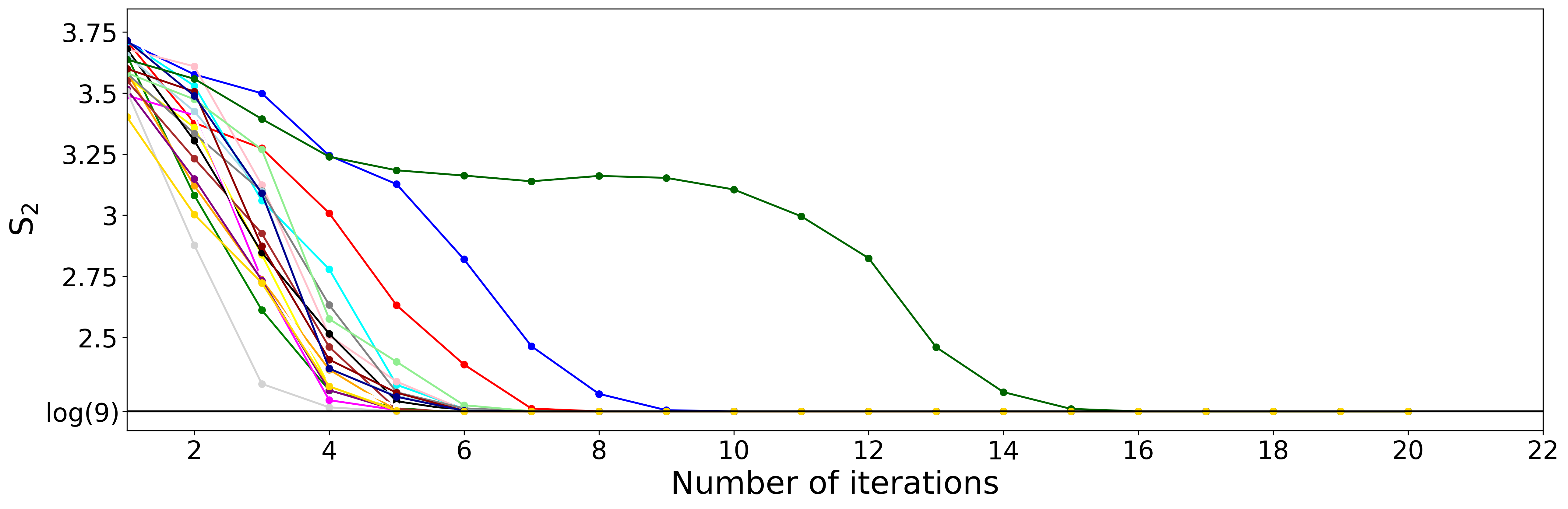}
\caption{Convergence of the algorithm for 20 random $\ame(4,3)$ states is shown, with each state represented by a different color. The decomposition entropy at each iteration step is plotted. The plot demonstrates convergence of $\enttwo$ to the theoretical value $\log(9)$. }
\label{fig:ame_43_convergence}
\end{figure*}

\subsection{Minimal decomposition entropy for $q=\infty$}
To calculate the minimal decomposition entropy $\entmininf$ of a given state $\ket{\Psi}$, we employ the seesaw algorithm\cite{streltsovSimpleAlgorithmComputing2011} and use the relation in Eq.~\ref{eq:sinf_gme}. The seesaw algorithm, originally proposed to find the GME, is used to identify the nearest fully separable state $\ket{\phi_{\text{sep}}}$ to an $N$-partite state $\ket{\Psi}$. The overlap $|\braket{\Psi |\phi_{\text{sep}}}|^2$ is then used to determine $\entmininf$. The relationship between the minimal decomposition entropy $\entmininf$ and the GME is provided in Eq.~\ref{eq:gme_sinfinity}, highlighting that studying $\entmininf$ is equivalent to investigating the GME. The GME has been studied extensively, and a recent review of this can be found in \cite{weinbrennerQuantifyingEntanglementGeometric2025}.

The distributions of $\entinf$ and $\entmininf$ for qubit, qutrit, and ququad systems are shown in Fig.~\ref{fig:mde_s_inf}, with the corresponding ensemble averages listed in Table~\ref{tab:average}. It is observed that, for the cases considered, typical AME states exhibit higher $\entmininf$ compared to Haar random states. In other words, in terms of both GME and minimal decomposition entropy, AME states are more entangled than Haar random states.

The maximal values of the GME and minimal decomposition entropy $\entmininf$ have been previously investigated \cite{steinbergFindingMaximalQuantum2024,enriquezMinimalRenyiIngarden2015,enriquezMaximallyEntangledMultipartite2016}. In Table~\ref{tab:max}, we present the maximum values obtained from the ensembles of Haar random states and AME states considered in our study. We did not find any new states with $\entmininf$ exceeding the previously known maximum values. However, a new four-ququad AME state has been identified with $\entmininf = 2.985$, which is higher than the previously reported value of $\entmininf = \log(16) \approx 2.772$ \cite{steinbergFindingMaximalQuantum2024}. The state is provided in \cite{N_Minimal_decomposition_entropy}.

\subsection{Minimal decomposition entropy for finite $q>1$}

The algorithm proposed in this work can be used to compute $\entmin$ for other values of $q>1$. To demonstrate this, we plot $\ent$ and $\entmin$ for a three qutrit Haar random state and an $\ame(3,3)$ state in Fig.~\ref{fig:varying_q}, with $1.5 \leq q \leq 100$ in increments of $0.5$. The value of the minimal decomposition entropy at $q=\infty$ is computed independently using the seesaw algorithm. As $q$ becomes larger, it is observed that the minimal decomposition entropy converges to the value obtained through the seesaw algorithm in each case.

\begin{figure}
\centering
\includegraphics[width=\columnwidth]{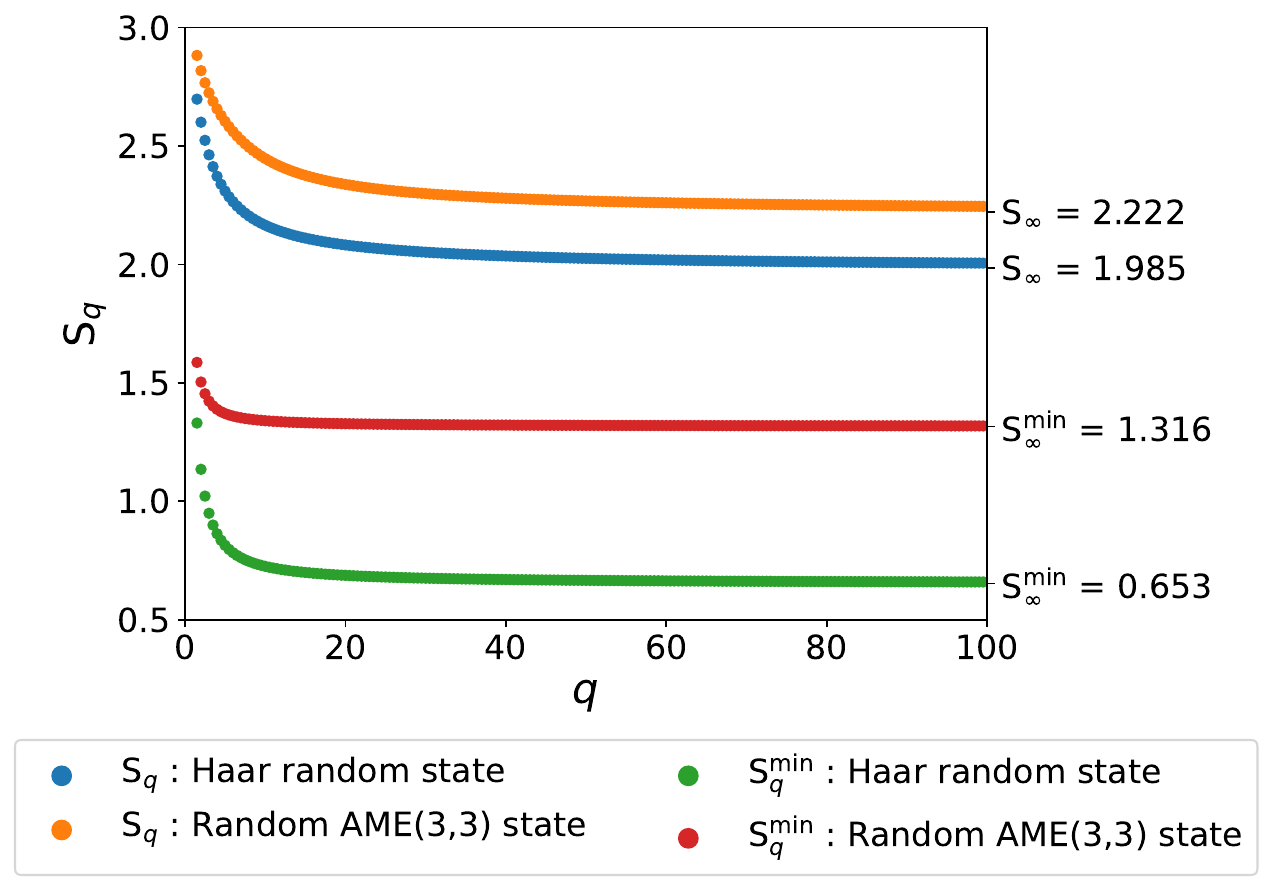}
\caption{Minimal decomposition entropy $\ent$ and $\entmin$ as a function of $q$ for a Haar random three qutrit state and a random $\ame(3,3)$ state plotted for $q>1$ with a step size 0.5. The values at $q=\infty$ are marked on the right-hand side.}
\label{fig:varying_q}
\end{figure}

\subsection{Optimal representation of AME states}
The state with minimum decomposition entropy, $\entmintwo$, provides an optimal representation in the sense that the entropy is minimized in the corresponding product basis, often yielding a simpler and sparser description of the state. The algorithm presented in this work identifies a product basis that achieves this minimum. It is important to note that this optimal form is not unique: Multiplying the state by local diagonal unitaries or applying local permutation matrices, for example, does not change the decomposition entropy. In principle, the Hartley entropy is the most direct measure of sparsity, as it simply counts the number of nonzero coefficients; minimizing it therefore yields the representation with truly minimal support within a given LU-equivalence class. However, the Hartley entropy is  difficult to optimize in practice, and for this reason one typically minimizes the Shannon entropy or decomposition entropies for other values of $q$.

We provide several examples illustrating how the algorithm for computing $\entmintwo$ yields a simpler and sparser representation of an AME state, often with significantly reduced support compared to the original form. The first example, shown in Fig.~\ref{fig:ame_43_example}, displays a 2-unitary matrix corresponding to a random $\ame(4,3)$ state alongside the final matrix obtained after applying the algorithm. In this case, the original state has support $81$, whereas the transformed state attains minimal support, as expected.

\begin{figure}
\centering
\begin{tabular}{cc}
\includegraphics[height=0.14\textheight]{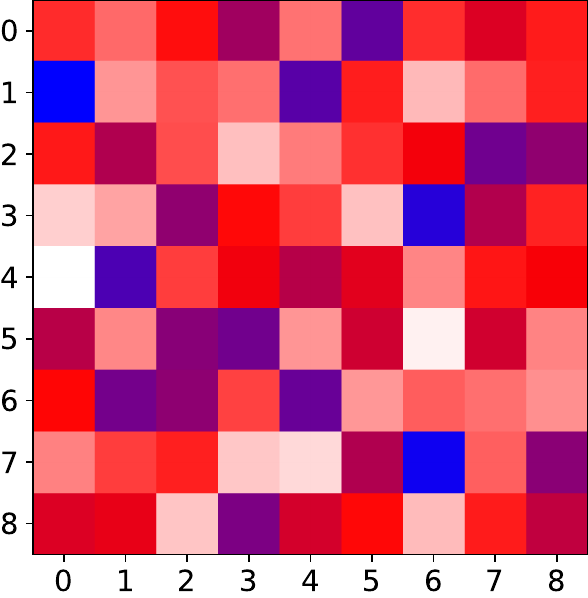}&
\includegraphics[height=0.14\textheight]{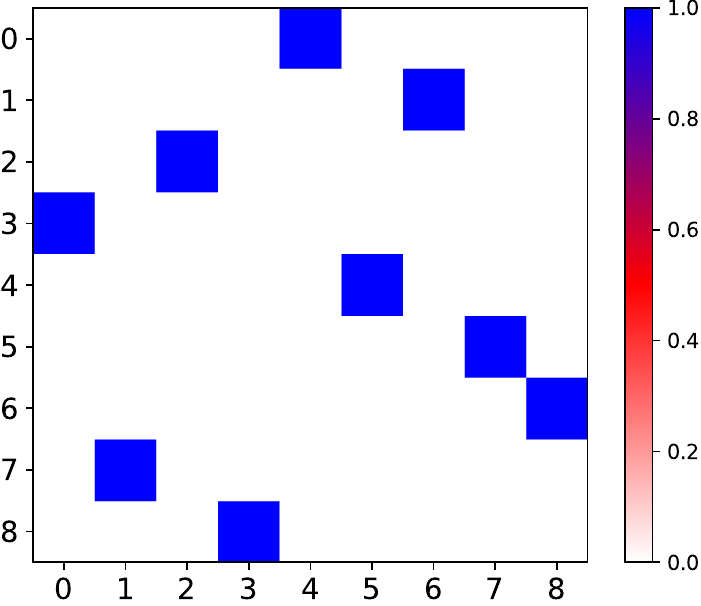}\\
(a) $\enttwo = 3.872609$  & (b) $\entmintwo = 2.197224$
\end{tabular}
\caption{(a) 2-unitary matrix corresponding to a random $\ame(4,3)$ state. (b) 2-unitary matrix corresponding the state obtained by applying the algorithm. Only the absolute values are shown in both cases. The values of entropy are given, and the $\entmintwo$ value matches the lower bound $\log(9)$ up to five decimal places for $\ame(4,3)$ states.}
\label{fig:ame_43_example}
\end{figure}

As another example, we consider the $\ame(4,4)$ state $\ket{O_{16}}$, corresponding to the 2-unitary matrix $O_{16}$ given in \cite{ratherConstructionLocalEquivalence2022}. Depictions of the original state $\ket{O_{16}}$ and the state obtained after applying our algorithm are shown in Fig.~\ref{fig:O16}. In this case, the final state has support $28$, compared to $64$ for the original. The algorithm-derived state can be further simplified by removing complex phases via local diagonal unitaries and by applying local permutations, yielding an even simpler representation. All of these versions—the original, the algorithm output, and the fully refined state—are provided in Appendix~\ref{app:o16}. Additional examples, including the cases of $\ame(4,6)$, $\ame(5,2)$, and $\ame(6,2)$, are presented in the appendix \ref{app:optimal_rep}. The quantum states and the local unitary matrices found using the algorithm are provided in \cite{N_Minimal_decomposition_entropy}.

\begin{figure}
\centering
\begin{tabular}{cc}
\includegraphics[height=0.14\textheight]{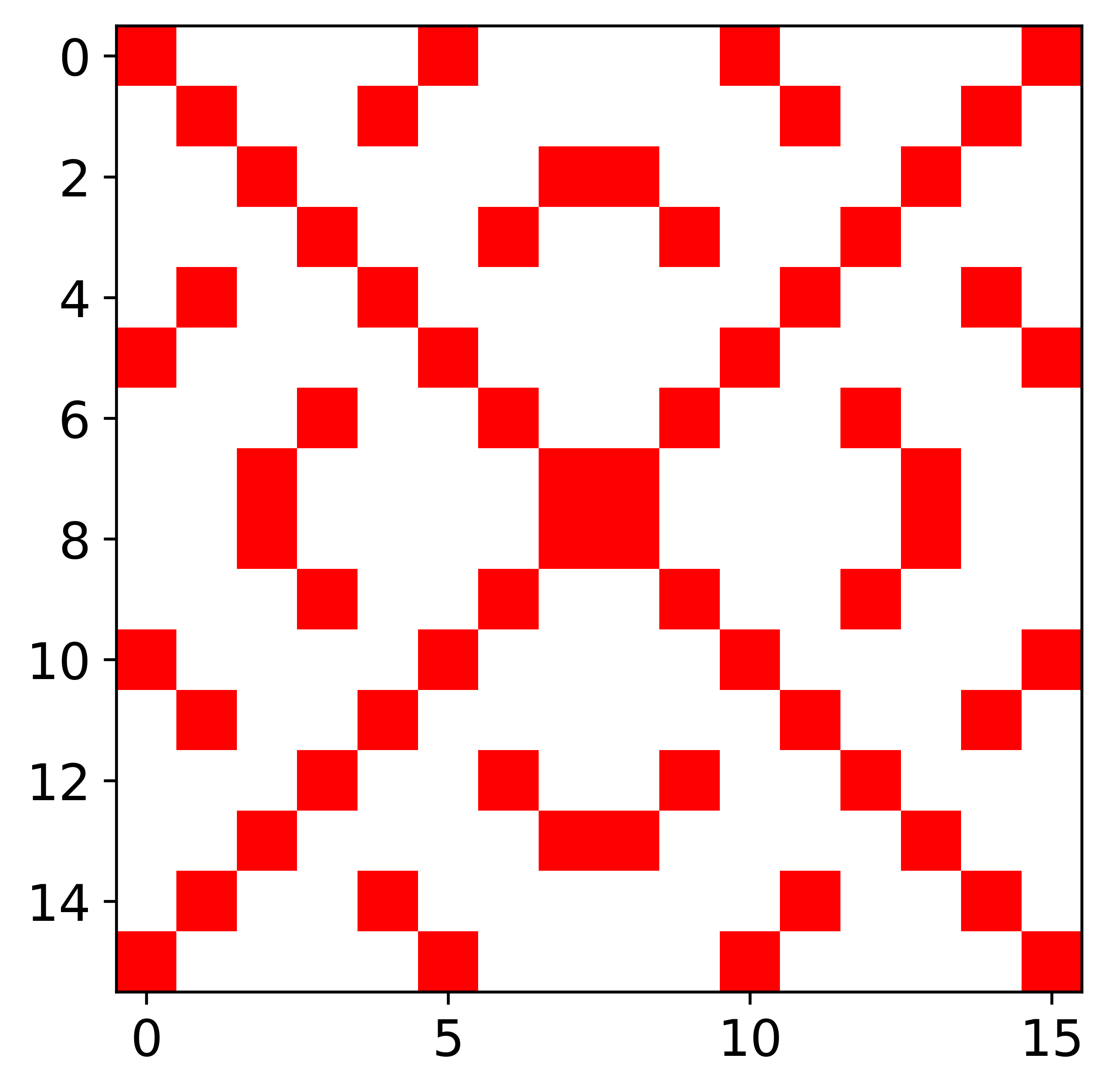} & \includegraphics[height=0.14\textheight]{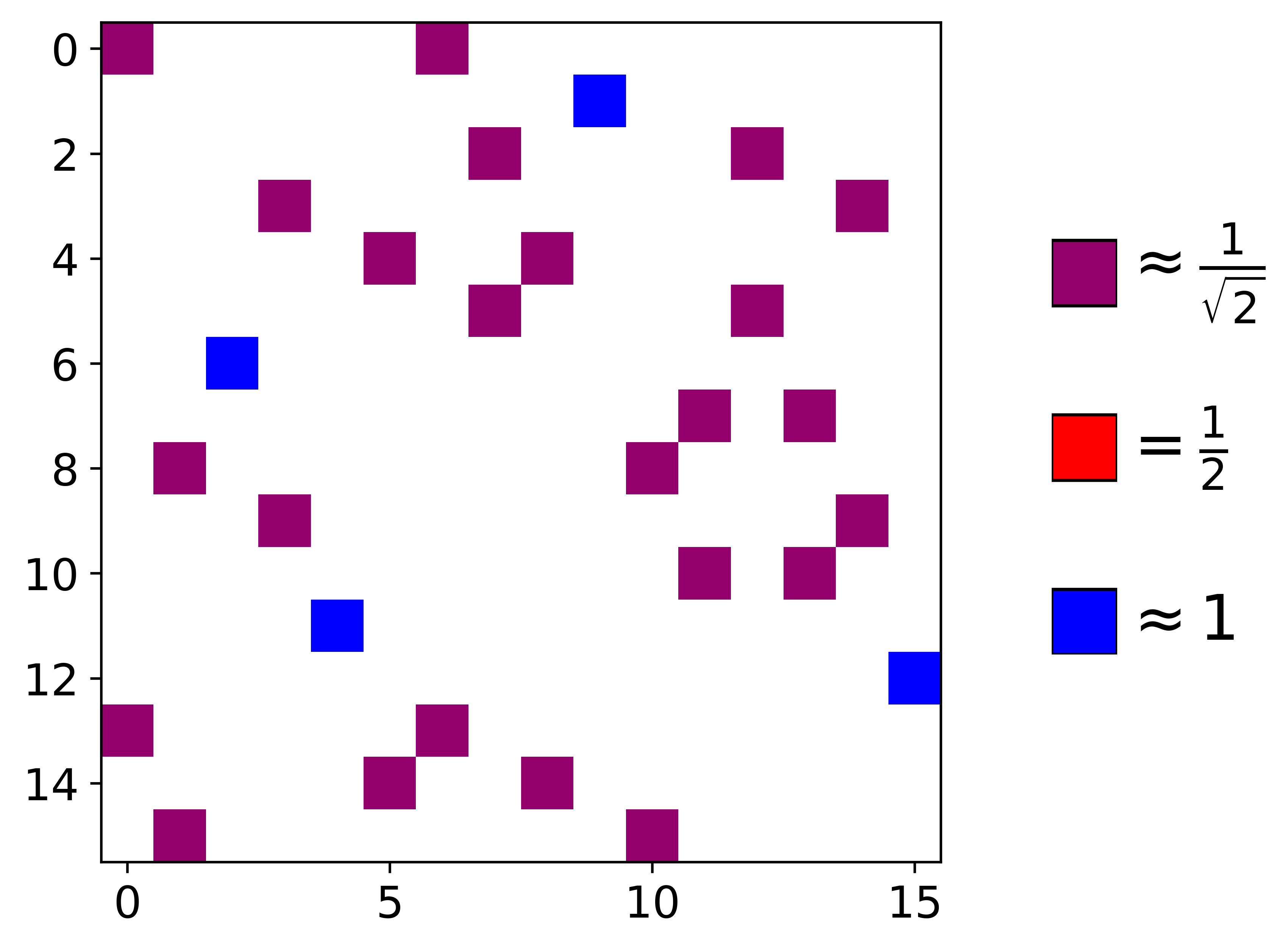} \\
(a) $\enttwo = 4.15888$ & (b) $\entmintwo = 3.24259$
\end{tabular}
\caption{(a) 2-unitary matrix $O_{16}$. (b) 2-unitary matrix corresponding the state obtained by applying the algorithm. Only the absolute values are shown in each cases. The values of entropy are given. The support of the final state is 28 compared to 64 of the original state.}
\label{fig:O16}
\end{figure}

Obtaining a simpler representation of an AME state is useful for determining whether the state is genuinely quantum or whether it can be constructed from a classical combinatorial structure. It is known that minimal-support AME states can be constructed from orthogonal arrays \cite{goyenecheGenuinelyMultipartiteEntangled2014}. Since the minimal decomposition entropy $\entmintwo$ attains its lower bound only for minimal-support states, any state that achieves this minimum is therefore not genuinely quantum. In the cases we examined, this implies that there are no genuinely quantum AME states for three qubits or for four qutrits. In contrast, for the $\ame(3,3)$ and $\ame(4,4)$ families, we find that most states in the ensemble are genuinely quantum. 

\section{Conclusions}
{
In this work, we examined the minimal decomposition entropy for a family of highly entangled multipartite states known as absolutely maximally entangled (AME) states. The decomposition entropy $\ent$ provides a measure of how localized a quantum state is within a given basis. When minimized over all local unitaries, it serves as an entanglement measure and yields an optimal representative of the state within its LU-equivalence class, making it a useful tool for exploring LU equivalence. For $q=\infty$, the minimal decomposition entropy is directly related to the geometric measure of entanglement.
}

An algorithm was introduced to determine the minimal decomposition entropy $\entmin$ for finite $q>1$, and the known seesaw algorithm was employed to compute $\entmininf$. We investigated AME states of qubits, qutrits, and ququads for the cases $q=2$ and $q=\infty$, and compared their minimal decomposition entropy against Haar-random states serving as a benchmark for generic multipartite entanglement. For three qutrits and four ququads, we observed that $\entmintwo$ is larger for Haar random states than for AME states, indicating that AME states can be more localized in these scenarios. In both ensembles, typical AME states were found to have larger $\entmininf$, and thus larger geometric entanglement, than their Haar random counterparts. We also identified new states exhibiting higher values of both $\entmintwo$ and $\entmininf$. Furthermore, we demonstrated that the algorithm is effective in obtaining optimal representations of AME states and in determining whether a given AME state is genuinely quantum -- that is, not realizable from a classical combinatorial design. For the $\ame(3,3)$ and $\ame(4,4)$ families, we found that most states are genuinely quantum.

The approach and algorithm developed in this work are not limited to AME states and can be applied to general multipartite quantum systems. While our analysis has focused on relatively small values of $N$ and $d$, the method can be readily extended to larger systems and higher-dimensional subsystems. Obtaining optimal or sparser representations of such states can provide valuable insights into their entanglement structure and other key properties. However, because the algorithm relies on a greedy optimization strategy, it may fail to converge to the global minimum and can become trapped in local minima, particularly for systems with many parties or higher local dimensions. This limitation underscores the need for further analytical and algorithmic developments. 

The theoretical upper bounds for the minimal decomposition entropy remain unknown for AME states and, more broadly, for many multipartite systems with general $N$ and $d$. The problem of finding maximally entangled states with respect to the minimal decomposition entropy remains largely open. { Whether the minimal decomposition entropy introduced here could also serve as a useful feature in machine learning-driven entanglement classification \cite{vintskevichClassificationFourqubitEntangled2023} remains an interesting question for future investigation.}

\section*{Data Availability}
Certain data and Python code that support the findings of this study are openly accessible on GitHub under the GNU General Public License v3.0 \cite{N_Minimal_decomposition_entropy}. Due to the large size of the remaining datasets, they are not publicly available. However, the code and algorithms required to generate and verify these datasets are provided in the repository.
\section*{Acknowledgement}
N.R. acknowledges funding received from the Center for Quantum Information, Communication, and Computing, IIT Madras, during his time at the institute, and thanks Arul Lakshminarayan for valuable discussions and insightful comments on the manuscript.
\appendix
\section{Principal component analysis using $L_p$-norm for complex matrices}\label{app:complexlppca}
Let $X = [x_1 x_2 \cdots x_N]$ be a $d \times N$ complex matrix. Define the function $F_p: \mathbb{C} \to \mathbb{R}$ such that
\begin{align}
\begin{aligned}
F_p(W) = \lVert W^\dagger X \rVert_p^p  = \sum_{i=1}^N \sum_{j=1}^m | w_j^\dagger x_i |^p,
\end{aligned}
\end{align} 
where $W$ is a $d\times m$ complex matrix with $W^\dagger W = \mathbb{I}_m$. Here $w_j$ denotes the $j$-th column of $W$ and $\mathbb{I}_m$ denotes the $m\times m$ identity matrix. The $L_p$ norm is convex for $p>1$. Hence, the function $F_p(W)$ is convex for $p>1$ since it is a sum of convex functions. The maximization problem is to find a $W$ such that $F_p(W)$ is maximum:

\begin{align}
\begin{aligned}
\text{(Complex $\lppca$)} ~~~ \max_{ W \in \mathbb{C}^{ d\times m}, W^\dagger W = I_m  } F_p(W).
\end{aligned}
\end{align}
For $p=2$, the problem becomes
\begin{align}
\begin{aligned}
F_2(W)  = \sum_{i=1}^N \sum_{j=1}^m | w_j^\dagger x_i |^2 = \tr \rl{ W^\dagger X X^\dagger W }.
\end{aligned}
\end{align}
The global optimal solution is found by solving $  X X^\dagger W = W D$, where $D$ is a diagonal matrix containing eigenvalues of $ X X^\dagger $. For other values of $p>1$, finding a solution is difficult. An algorithm is proposed by extending the methods given in \cite{kwakPrincipalComponentAnalysis2014} for complex data.

\begin{figure*}[h]
\centering
\begin{tabular}{cc}
\includegraphics[height=0.3\textheight]{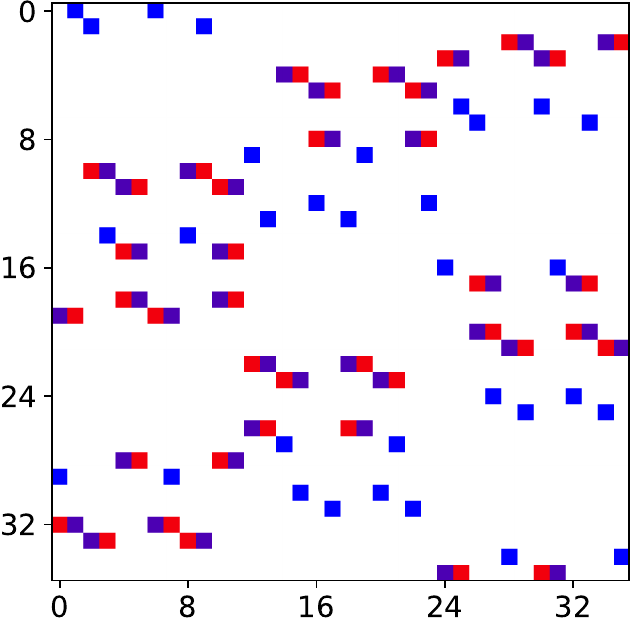} & \includegraphics[height=0.3\textheight]{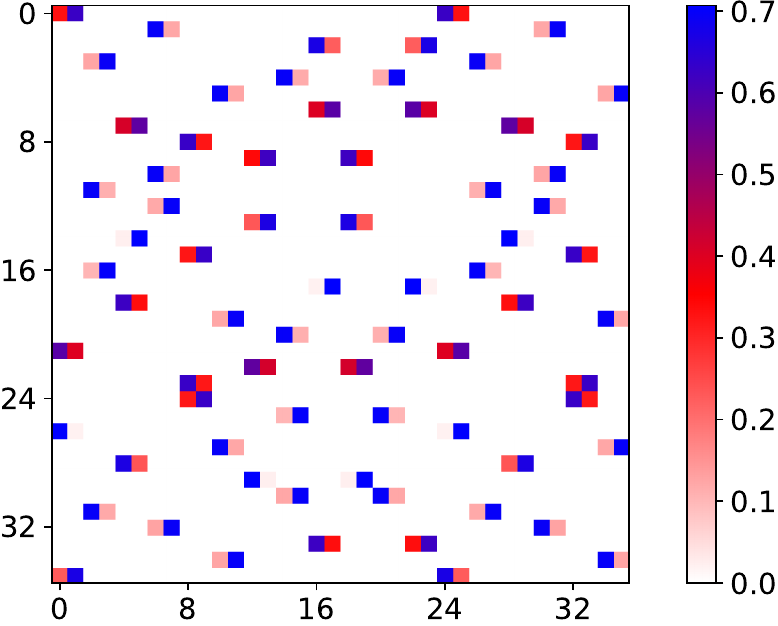} \\
(a) $\enttwo = 4.527$  & (b) $\entmintwo = 4.463$
\end{tabular}
\caption{(a) The 2-unitary matrix corresponding to the Golden AME state \cite{ratherThirtysixEntangledOfficers2022}, presented in its simplified form. (b) The 2-unitary matrix for the state obtained after applying the algorithm. Only the absolute values are shown, and the entropy values are provided. The support of the final state is 144, compared to 112 for the original, indicating an increase in support despite the lower entropy. This is expected, as the minimum value of $\entmintwo$ may not yield the state with minimal support within a given LU equivalence class.}
\label{fig:ame46}
\end{figure*}
\begin{figure*}
\centering
\begin{tabular}{cc}
\includegraphics[height=0.15\textheight]{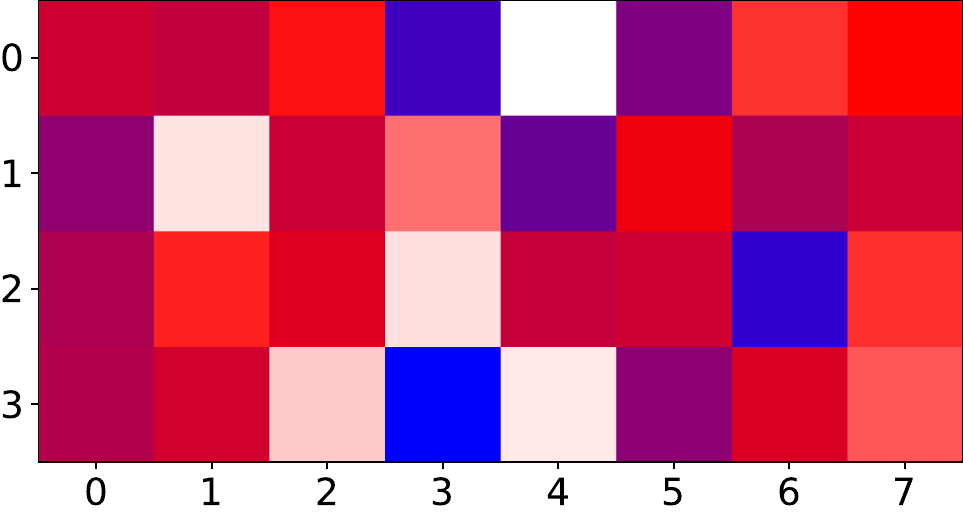} & \includegraphics[height=0.15\textheight]{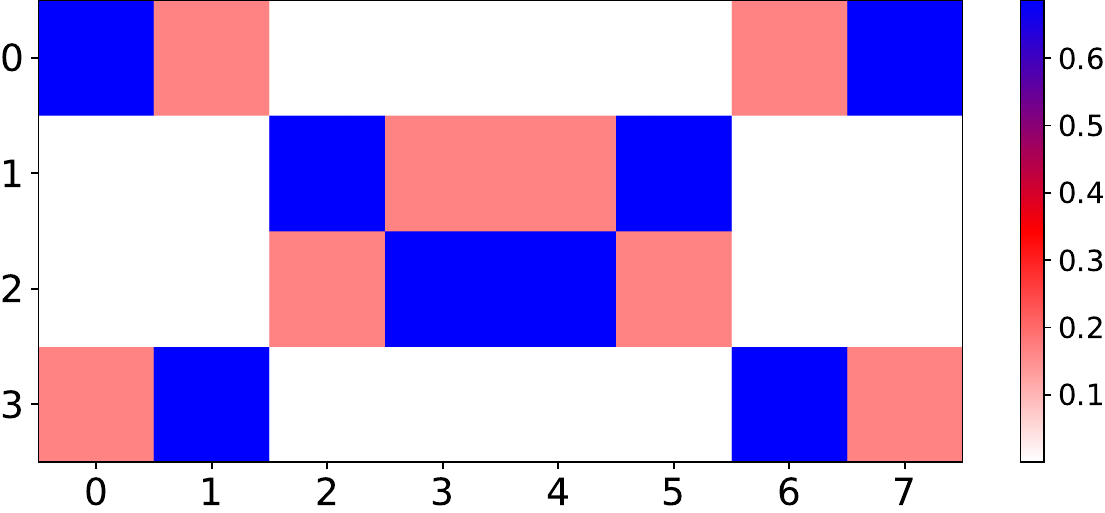} \\
(a) $\enttwo = 3.151$  & (b) $\entmintwo = 2.192$
\end{tabular}
\caption{(a) The isometry corresponding to a random $\ame(5,2)$ state. (b) The isometry corresponding to the same state obtained by applying the algorithm. Only the absolute values of the matrix elements are shown in each case, with entropy values provided. The support of the final state is 16, compared to 32 for the original state.}
\label{fig:ame52}
\end{figure*}
\begin{figure*}
\begin{tabular}{cc}
\includegraphics[height=0.2\textheight]{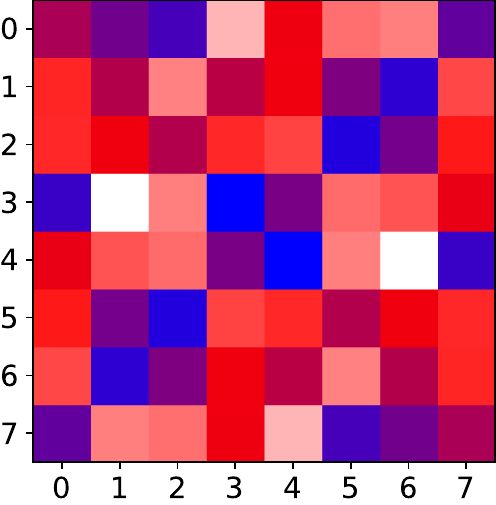} & \includegraphics[height=0.2\textheight]{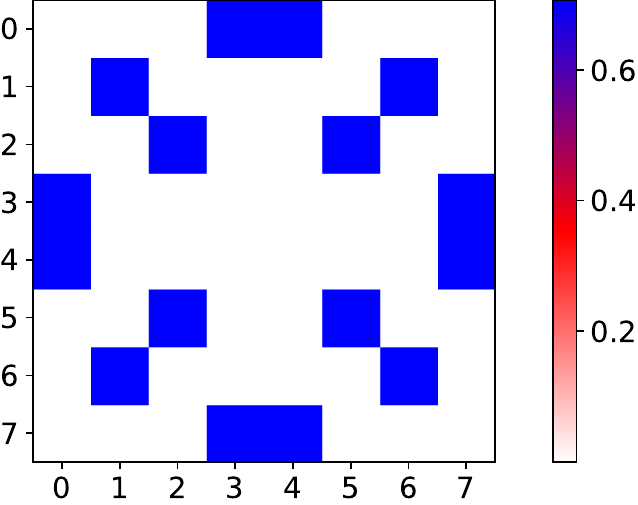} \\
(a) $\enttwo = 3.739$  & (b) $\entmintwo = 2.772$
\end{tabular}
\caption{(a) The 3-unitary matrix corresponding to a random $\ame(6,2)$ state. (b) The 3-unitary matrix for the same state after applying the algorithm. Only the absolute values of the matrix elements are shown, and the entropy values are given. The support of the final state is 16 compared to 64 of the original state.}

\label{fig:ame62}
\end{figure*}

Define the Lagrangian of the problem
\begin{align}
\begin{aligned}
L(W, \Lambda) &= \lVert W^\dagger X \rVert_p^p + \tr( \Lambda ( W^\dagger W - I_m ))\\
& = \sum_{i=1}^N \sum_{j=1}^m | w_j^\dagger x_i |^p + \sum_{k,l} \lambda_{k,l} \rl{ w_k^\dagger w_l - \delta_{k,l} },
\end{aligned}
\end{align}
where $ \Lambda $ is an $m\times m$ Hermitian matrix. The necessary condition for the optimal solution can be obtained by setting the derivative of the Lagrangian to zero:
\begin{align}
\begin{aligned}
\frac{dL}{dW^*} = 0.
\end{aligned}
\end{align}
Note that $\frac{dL}{dW^*} = (\frac{dL}{dW})^* $. Define the gradient
\begin{align}
\begin{aligned}
\nabla_{W^* } = \frac{dF_p(W)}{dW^*} = [  \nabla_{w_1^*}~\nabla_{w_2^*}~\cdots~\nabla_{w_m^*}]
\end{aligned}
\end{align}
where  \begin{align}
\begin{aligned}
\nabla_{w_k^*} &=  \frac{\partial }{\partial w_k^*} \rl{ \sum_{i=1}^N \sum_{j=1}^m | w_j^\dagger x_i |^p }\\
& =  \sum_{i=1}^N \frac{p}{2} ( w_k^\dagger x_i )^{p/2-1} ( x_i^\dagger w_k )^{p/2} x_i .
\end{aligned}
\end{align} 

Also,
\begin{align}
\begin{aligned}
\frac{\partial }{\partial w^*_q} \rl{ \sum_{k,l} \lambda_{k,l} \rl{ w_k^\dagger w_l - \delta_{k,l} } } 
& =  \sum_{k,l} \lambda_{k,l} \rl{ w_l \delta_{k,q} } \\&= \sum_{l } \lambda_{q,l}  w_l .
\end{aligned}
\end{align}
Therefore, the condition becomes
\begin{align}
\begin{aligned}
\frac{dL}{dW^*} = \nabla_{W^*} +   W \Lambda = 0.
\end{aligned}
\end{align}

In each iteration, the matrix $W$ is chosen by solving the optimization problem
\begin{align}
\begin{aligned}\label{eq:optimization}
W' = \argmax_{ Q \in \mathbb{C}^{d\times m}, ~~Q^\dagger Q = \mathbb{I}_m  } \re \rl{ \tr \rl{ Q^\dagger \nabla_{W^*} }}.
\end{aligned}
\end{align}
The idea is to find a $W$ as close to $\nabla_{W^*} $ such that the constraint $  W^\dagger W = \mathbb{I}_m  $ is satisfied \cite{kwakPrincipalComponentAnalysis2014}. The solution to the above optimization problem is given by
\begin{align}
\begin{aligned}
W' = U [I_m| \pmb{0}] V^\dagger,
\end{aligned}
\end{align}
where the unitary matrices $U$ and $V$ are obtained from the singular value decomposition $\nabla_{W^*} = U D V^\dagger$. The iterative update ensures that $F_p(W') \geq F_p(W)$. This can be shown as follows:

The function $ F_p(W)$ is convex for $p>1$. Therefore, the first order convexity condition implies \cite{sayedAdaptationLearningOptimization2014}
\begin{align}
\begin{aligned}
F_p(W') - F_p(W) \geq 2 \re \rl{ \tr ( ( {W'}^\dagger - W^\dagger  ) \nabla_{W^*}    ) }
\end{aligned}
\end{align}
Since $W'$ is a solution of the optimization problem in Eq.~ \ref{eq:optimization}, this gives $ \re \rl{ \tr (  {W'}^\dagger  \nabla_{W^*}    ) } \geq \re \rl{ \tr (  {W}^\dagger  \nabla_{W^*}    ) }  $. Therefore,
\begin{align}
\begin{aligned}\label{eq:convergence}
F_p(W') - F_p(W)  \geq 0.
\end{aligned}
\end{align}
The algorithm is given below:\\
\rule{\columnwidth}{1pt}
Algorithm : Complex $\lppca$ \\
\rule{\columnwidth}{1pt}\\
Initialize $W_0$ such that $W_0^\dagger W_0 = I_m$. For $i=0,1,...$
\begin{itemize}
\item[1)] Compute $\nabla_{W^*_i}$
\item[2)] SVD : $\nabla_{W^*_i} = U_i D V_i^\dagger  $
\item[3)] $W_{i+1} \longleftarrow U_i [I_m| \pmb{0}] V_i^\dagger $
\end{itemize}
\rule{\columnwidth}{1pt}\\
The iterative update in Eq.~\ref{eq:optimization} ensures that the algorithm converges. 

\section{Algorithm to generate k-uniform and AME states} \label{app:ame_algorithm}
To construct  $k$-uniform and absolutely maximally entangled (AME) states, we develop an algorithm inspired by the methods introduced in \cite{ratherCreatingEnsemblesDual2020} for generating dual-unitary and 2-unitary matrices. Consider a complex matrix $A \in M_{m,n}(\mathbb{C})$ with $m \leq n$ and normalization condition $\tr\rl{ A^\dagger A} =m$. Let \( A = UDV\) be the singular value decomposition of \( A \), where \( D_{ij} = s_i \delta_{ij} \) for \( i = 1, \dots, m \) and \( j = 1, \dots, n \), with \( s_i \) representing the singular values of \( A \). The matrices \( U \) and \( V \) are unitary, with dimensions \( m \times m \) and \( n \times n \), respectively.
 We define the map
\begin{align}
\begin{aligned}
\neari (A)  = U [I_m | 0]V,
\end{aligned}
\end{align}
where $[I_m | 0]$ represents the $m\times n$ matrix obtained by concatenating the identity matrix $ I_m $ with an $m \times (n-m) $ zero block. The matrix $\neari (A)$ gives the isometry closest to 
$A$ in the Frobenius norm, corresponding to the solution of the a rectangular variant of the orthogonal Procrustes problem \cite{schönemannGeneralizedSolutionOrthogonal1966}.

The nuclear norm of a matrix is defined as the sum of its singular values,
\begin{align}
\begin{aligned}
\rVert A\lVert_* = \sum_{i=1}^m \sigma_i.
\end{aligned}
\end{align}
It can equivalently be expressed as
\begin{align}
\begin{aligned}
\rVert A\lVert_* = \tr\rl{ A^\dagger \neari (A) }.
\end{aligned}
\end{align}
The nuclear norm $\rVert A\lVert_* $ attains its maximum value $m$ when $A$ is an isometry. To see this, note that the normalization condition $\tr\rl{ A^\dagger A} =m$ implies $\sum_{i=1}^m s_i^2 = m$. Using Cauchy-Schwarz inequality, \[\sum_{i=1}^m s_i \leq \sqrt{m} \sqrt{ \sum_{i=1}^m s_i^2 } = m,\] and equality holds if and only if $s_i=1$ for all $i$, that is, when $A$ is an isometry. 

Additionally, 
\begin{align}\label{eq:nearestiso}
\begin{aligned}
\rVert A\lVert_* = \tr\rl{ A^\dagger \neari (A) } \geq  \re \rl{\tr \rl{A^\dagger Q }}
\end{aligned}
\end{align}
for any isometry $Q \in M_{m,n}(\mathbb{C})$. Indeed,
\begin{align*}
\begin{aligned}
\re \rl{\tr \rl{A^\dagger Q }} &=  \re \rl{\tr \rl{V^\dagger DU^\dagger Q }} \\
& =  \re \rl{\tr \rl{DQ'}} ~~~~( Q' =U^\dagger Q V^\dagger )\\
& = \re\rl{ \sum_{i=1}^m s_i Q'_{ii}}.
\end{aligned}
\end{align*}
Since
\begin{align*}
\begin{aligned}
\re\rl{ \sum_{i=1}^m s_i Q'_{ii}} \leq | \rl{ \sum_{i=1}^m s_i Q'_{ii}}  | \leq \rl{ \sum_{i=1}^m s_i |Q'_{ii}|}
\end{aligned}
\end{align*}
and for an isometry $Q'$, one has $|Q'_{ii}| \leq 1$, it follows that
\begin{align}
\begin{aligned}
\re \rl{\tr \rl{A^\dagger Q }} \leq \sum_{i=1}^m s_i = \tr\rl{ A^\dagger \neari (A) }.
\end{aligned}
\end{align}

Consider an $N$-partite state $\ket{\Psi}$, and define its associated $d^k \times d^{N-k}$ matrix $A$ as given in Eq~\ref{eq:matricization}, where $k \leq \lfloor N/2 \rfloor$. For AME states $k$ is chosen as $k = \lfloor N/2 \rfloor$. We introduce the objective function
\begin{align}
\begin{aligned}
f(A) = \mathcal{N}  \sum_{\sigma \in \perm_N^k }  \rVert A^{R_\sigma} \lVert_* ,
\end{aligned}
\end{align}
 where $\mathcal{N}  = \frac{k! (N-k)!}{d^k N!} $ and $ \perm_N^k  = \perm_N/ ( \perm_k \times \perm_{N-k} )$.

The norm $\rVert A^{R_\sigma} \lVert_*$ reaches its maximum value if and only if $A^{R_\sigma}$ is an isometry. Consequently, $f(A)$ attains its maximal value of $1$ when all the matrices $A^{R_\sigma}$ corresponding to permutations $\sigma \in \perm_N^k$ are isometries -- equivalently, when the state associated with the matrix $A$ is $k$-uniform. The normalization factor $\mathcal{N} $ is chosen such that maximum value of $f(A)$ is $1$; this choice accounts for the size of the set $\perm_N^k$, which contains $N!/(k! (N-k)!)$ permutations, and the fact that the nuclear norm of a $d^k \times d^{N-k}$ isometry equals $d^k$.

We find an algorithm that maximize $f(A)$ and it is given now. \\
\noindent \rule{\columnwidth}{1pt}
Algorithm to generate $k$-uniform states\\ 
\rule{\columnwidth}{1pt}
Initialize random $d^k \times d^{(N-k)}$ isometry $A_0$.
\begin{itemize}
\item[] For $i=1,2,...$\\
$\qquad A_{i} =  \mathcal{N}  \displaystyle\sum_{\sigma \in \perm_N^k }   \rl{ \neari \rl{  A_{i-1}^{R_\sigma}  } }^{ R_{ \sigma^{-1}} } $
\item[] Repeat until convergence of $f(A_i)$
\end{itemize}
\rule{\columnwidth}{1pt}

The iterative update \[A_{i+1} = \mathcal{N}  \sum_{\sigma \in \perm_N^k }   \rl{ \neari \rl{  A_i^{R_\sigma}  } }^{ R_{ \sigma^{-1}} } \] ensures that $f(A_i)$ forms a non-decreasing sequence. The objective function $f(A_i)$ can be expressed as 
\begin{align}
\begin{aligned}
f(A_i) =&\mathcal{N} \sum_{\sigma \in \perm_N^k }    \tr\rl{ \rl{A_i^{R_\sigma}}^\dagger \neari \rl{  A_i^{R_\sigma}  }  } \\
&=  \mathcal{N}   \tr\rl{ A_i^\dagger \sum_{\sigma \in \perm_N^k }   \rl{ \neari \rl{  A_i^{R_\sigma}  } }^{ R_{ \sigma^{-1}} }  }\\
&= \tr\rl{ A_i^\dagger A_{i+1}}.
\end{aligned}
\end{align}
The identity $\tr\rl{   {A^{R_{\sigma}}}^\dagger B } = \tr\rl{   A^\dagger B^{R_{\sigma^{-1}}  }  }$ is used in this derivation. Similarly,
\begin{align*}
\begin{aligned}
f(A_{i-1}) &=  \tr\rl{ A_{i-1}^\dagger A_{i}} 
=  \tr\rl{ A_{i}^\dagger A_{i-1}} \\
&  =  \mathcal{N} \sum_{\sigma \in \perm_N^k }    \tr\rl{ \rl{A_i^{R_\sigma}}^\dagger \neari \rl{  A_{i-2}^{R_\sigma}  }  }.
\end{aligned}
\end{align*}
From the relation stated in Eq.~\ref{eq:nearestiso}, it follows that
\[ \tr\rl{ \rl{A_i^{R_\sigma}}^\dagger \neari \rl{  A_i^{R_\sigma}  }  } \geq \tr\rl{ \rl{A_i^{R_\sigma}}^\dagger \neari \rl{  A_{i-2}^{R_\sigma}  }  },  \]
and hence $f(A_i) \geq f(A_{i-1})$. Although this iterative scheme guarantees that the objective function is non-decreasing, convergence to the global maximum is not guaranteed, and the procedure may become trapped in local maxima. To mitigate this, several runs with randomly initialized seed isometries are typically performed.

\section{$\ame(4,4)$ states $\ket{O_{16}}$ and $\ket{O_{16}'}$} \label{app:o16}

The \(\ame(4,4)\) state \(\ket{O_{16}}\), corresponding to the 2-unitary matrix \(O_{16}\), the state obtained from the algorithm \(\ket{O^{\text{alg}}_{16}}\), which corresponds to the minimal entropy value, and the refined state \(\ket{O^{\text{ref}}_{16}}\), obtained by removing complex phases and applying local permutations, are
{
\begin{widetext}
\begin{align}
\ket{O_{16}}=&\frac{1}{8}(-\ket{1,1,3,3}-\ket{1,1,4,4}-\ket{1,2,2,1}-\ket{1,2,3,4}
-\ket{1,2,4,3}-\ket{1,3,1,3}-\ket{1,3,3,1}-\ket{1,3,4,2 }   \nonumber \\&
-\ket{1,4,1,4}-\ket{2,1,2,1}-\ket{2,2,1,1}-\ket{2,2,4,4}
-\ket{2,3,1,4}-\ket{2,3,3,2}-\ket{2,3,4,1}-\ket{2,4,1,3}   \nonumber \\&
-\ket{2,4,2,4}-\ket{2,4,4,2}-\ket{3,1,1,3}-\ket{3,1,2,4}
-\ket{3,1,3,1}-\ket{3,2,3,2}-\ket{3,3,1,1}-\ket{3,3,2,2}   \nonumber \\&
-\ket{3,3,3,3}-\ket{3,3,4,4}-\ket{3,4,1,2}-\ket{3,4,2,1}
-\ket{3,4,3,4}-\ket{4,1,1,4}-\ket{4,1,2,3}-\ket{4,1,3,2}   \nonumber \\&
-\ket{4,2,2,4}-\ket{4,2,3,1}-\ket{4,2,4,2}-\ket{4,3,3,4}
-\ket{4,4,2,2}-\ket{4,4,4,4}+\ket{1,1,1,1}+\ket{1,1,2,2}   \nonumber \\&
+\ket{1,2,1,2}+\ket{1,3,2,4}+\ket{1,4,2,3}+\ket{1,4,3,2}
+\ket{1,4,4,1}+\ket{2,1,1,2}+\ket{2,1,3,4}+\ket{2,1,4,3}   \nonumber \\&
+\ket{2,2,2,2}+\ket{2,2,3,3}+\ket{2,3,2,3}+\ket{2,4,3,1}
+\ket{3,1,4,2}+\ket{3,2,1,4}+\ket{3,2,2,3}+\ket{3,2,4,1}   \nonumber \\&
+\ket{3,4,4,3}+\ket{4,1,4,1}+\ket{4,2,1,3}+\ket{4,3,1,2}
+\ket{4,3,2,1}+\ket{4,3,4,3}+\ket{4,4,1,1}+\ket{4,4,3,3})\\
\ket{O^{\text{alg}}_{16}}=&(-0.0010+0.1768 i)\ket{3,3,3,4}+(0.1766-0.0079i)\ket{2,4,3,4}
+(-0.0091+0.1765 i)\ket{2,4,4,2}   \nonumber \\&+(-0.1760-0.0163i)\ket{3,1,3,3}
+(0.0180-0.1759 i)\ket{4,2,1,1}+(0.1759+0.0180 i)\ket{3,3,4,2}   \nonumber \\&+(0.2478-0.0328 i)\ket{2,3,1,3}
+(-0.0338+0.1735 i)\ket{4,4,3,3}+
(-0.2472-0.0373 i)\ket{3,4,2,1}
   \nonumber \\&
+(0.1711+0.0443 i)\ket{1,4,4,3}+(-0.0483+0.1701 i)\ket{2,1,2,2}+(0.0627+0.1653 i)\ket{4,2,2,3}
   \nonumber \\&
+(-0.2415-0.0645 i)\ket{1,2,3,2}+(-0.0654+0.1642 i)\ket{2,1,3,1}+(0.0675-0.1634 i)\ket{1,3,2,4}
   \nonumber \\&
+(-0.1628+0.0690 i)\ket{2,2,4,1}+(-0.1613+0.0724 i)\ket{3,2,1,4}+(-0.1609-0.0732 i)\ket{4,3,2,2}
   \nonumber \\&
+(0.0776-0.1588 i)\ket{3,1,1,2}+(-0.1569-0.0814 i)\ket{1,1,2,3}+(-0.0875+0.2342 i)\ket{4,1,4,4}
   \nonumber \\&
+(0.0876-0.1535 i)\ket{1,4,1,4}+(0.1526+0.0893 i)\ket{4,3,3,1}+(0.1503+0.0931 i)\ket{4,4,1,2}
   \nonumber \\&
+(-0.0984+0.1469 i)\ket{1,3,4,1}+(-0.1459+0.0998 i)\ket{2,2,2,4}+(-0.1039-0.1430 i)\ket{1,1,1,1}
   \nonumber \\&
+(0.1122+0.1366 i)\ket{3,2,4,3}\\
\ket{O^{\text{ref}}_{16}}= &\frac{1}{4}\bigg(
\ket{1,1,1,1}
+\ket{2,2,2,2}
+\ket{3,3,3,3}
+\ket{4,4,4,4}
+\frac{1}{\sqrt{2}}\ket{1,2,3,4}
-\frac{1}{\sqrt{2}} \ket{1,3,2,4}
-\frac{1}{\sqrt{2}}  \ket{2,3,4,1}
    \nonumber \\&
-\frac{1}{\sqrt{2}}  \ket{2,4,1,3}
-\frac{1}{\sqrt{2}}  \ket{4,2,3,1}
-\frac{1}{\sqrt{2}} \ket{4,3,1,2}
-\frac{i}{\sqrt{2}}  \ket{1,4,2,3}
-\frac{i}{\sqrt{2}} \ket{2,1,3,4})
-\frac{i}{\sqrt{2}}\ket{2,3,1,4})
    \nonumber \\&
-\frac{i}{\sqrt{2}} \ket{3,2,1,4}
+\frac{i}{\sqrt{2}}  \ket{1,2,4,3}
+\frac{i}{\sqrt{2}}  \ket{1,4,3,2}
+\frac{i}{\sqrt{2}}  \ket{2,4,3,1}
+\frac{i}{\sqrt{2}}  \ket{4,2,1,3}
+\frac{1}{\sqrt{2}}  \ket{1,3,4,2}
   \nonumber \\&
+\frac{1}{\sqrt{2}}  \ket{2,1,4,3}
+\frac{1}{\sqrt{2}}  \ket{3,1,2,4}
+\frac{1}{\sqrt{2}}  \ket{3,1,4,2}
+\frac{1}{\sqrt{2}}  \ket{3,2,4,1}
+\frac{1}{\sqrt{2}}  \ket{3,4,1,2}
+\frac{1}{\sqrt{2}}  \ket{3,4,2,1}
   \nonumber \\&
+\frac{1}{\sqrt{2}}  \ket{4,1,2,3}
+\frac{1}{\sqrt{2}}  \ket{4,1,3,2}
+\frac{1}{\sqrt{2}}  \ket{4,3,2,1} \bigg). 
\end{align}
\end{widetext}
}
This example illustrates a case where the state obtained after applying the algorithm has reduced support, yielding a sparser representation. Specifically, the support of the original state is $64$, whereas the support of the optimized state $\ket{O^{\text{alg}}_{16}}$ is $28$.

\section{Examples of optimal representations of AME states}
\label{app:optimal_rep}

In this appendix we present additional examples demonstrating how the algorithm can be used to obtain simpler and sparser representations of $\ame$ states. In particular, the cases of $\ame(4,6)$, $\ame(5,2)$, and $\ame(6,2)$ are considered. It is known that only a single LU-equivalence class of $\ame(6,2)$ states exists \cite{rainsQuantumCodesMinimum1999}. Although an infinite number of $\ame(5,2)$ states exist \cite{ramadasLocalUnitaryEquivalence2025}, they all lie within a single two-dimensional subspace, and every $\ame(5,2)$ state is LU-equivalent to an element of this subspace \cite{rainsQuantumCodesMinimum1999}. In all of these cases, the corresponding $\ame$ states are genuinely quantum. The examples are shown in Figs.~\ref{fig:ame46}--\ref{fig:ame62}, and the original states together with the local unitary matrices that map them to their final forms are provided in \cite{N_Minimal_decomposition_entropy}. 


\bibliographystyle{unsrt}
\bibliography{/home/izwez/Universe/Dropbox/Work/bibtex/global_references.bib,bibliography}

\end{document}